\documentclass{article}

\usepackage{latexsym}
\usepackage{amssymb,amsmath, amsthm, amsfonts}
\usepackage{graphicx}
\usepackage{color}

\makeatletter
\providecommand*{\Dashv}{%
  \mathrel{%
    \mathpalette\@Dashv\vDash
  }%
}
\newcommand*{\@Dashv}[2]{%
  \reflectbox{$\m@th#1#2$}%
}
\makeatother

\newtheorem{theorem}{Theorem}
\newtheorem{proposition}{Proposition}
\newtheorem{lemma}{Lemma}
\newtheorem{corollary}{Corollary}

\theoremstyle{definition}
\newtheorem{definition}{Definition}
\newtheorem{remark}{Remark}
\newtheorem{example}{Example}

\newcommand{\algLfin}{\textbf{\L}\mathbf{V}}
\newcommand{\domLfin}{{\L}V}

\newcommand{\ninv}{\mathord{\sim}}  
\newcommand{\logic}[1]{\mathrm{ #1}}
\newcommand{\alg}[1]{{\ensuremath{\boldsymbol{ #1}}}}

\newcommand{\Gsimleq}{\mathrm{G}\mystrut^{\text{\tiny $\leq$}}_{\sim}}
\newcommand{\Gnsimleq}{\mathrm{G}\mystrut^{\text{\tiny $\leq$}}_{n\sim}}
\newcommand{\Gasimleq}{\mathrm{G}\mystrut^{\text{\tiny $\leq$}}_{3\sim}}
\newcommand{\Gbsimleq}{\mathrm{G}\mystrut^{\text{\tiny $\leq$}}_{4\sim}}
\newcommand{\Gcsimleq}{\mathrm{G}\mystrut^{\text{\tiny $\leq$}}_{5\sim}}

\newcommand\mystrut{\rule{0pt}{7pt}}
\newcommand\mystrutsec{\rule{0pt}{12pt}}
\newcommand\mystrutsubsec{\rule{0pt}{10pt}}

\begin{document}

\title{
Degree-preserving G\"odel logics with an involution: intermediate logics and (ideal) paraconsistency\thanks{This paper is a humble tribute to our friend and colleague Arnon Avron and his outstanding contributions on nonclassical logics, proof theory and foundations of mathematics. \newline \mbox{} \newline {\bf This paper has been published in Ofer Arieli \& Anna Zamansky (eds.), {\em Arnon Avron on Semantics and Proof Theory of Non-Classical Logics}. Cham: Springer Verlag. pp. 107-139 (2021) }
}}
\author{Marcelo E. Coniglio$^1$, Francesc Esteva$^2$, Joan Gispert$^3$ and Lluis Godo$^2$ \\ \\
\small $^1$ Dept. of Philosophy - IFCH and  \\ \small  Centre for Logic, Epistemology and the History of Science, \\ \small University of Campinas, Brazil \\
\small {\tt coniglio@unicamp.br} \\
\small $^2$ Artificial Intelligence Research Institute (IIIA) - CSIC, Barcelona, Spain \\
\small {\tt \{esteva,godo\}@iiia.csic.es} \\
\small $^3$ Departament de Matem\`{a}tiques i Inform\`{a}tica, Universitat de Barcelona, Spain \\
\small {\tt jgispertb@ub.edu}
}
\date{}
\maketitle

\begin{abstract} In this paper we study intermediate logics between the logic $\Gsimleq$, the degree preserving companion of G\"odel fuzzy logic with involution $\rm{G}_\sim$ and classical propositional logic CPL, as well as the intermediate logics of their finite-valued counterparts $\Gnsimleq$. Although $\Gsimleq$ and  $\Gnsimleq$ are explosive w.r.t.\ G\"odel negation $\neg$, they are paraconsistent w.r.t.\ the involutive negation $\ninv$. We introduce the notion of saturated paraconsistency, a weaker notion than ideal paraconsistency, and we fully characterize the ideal and the saturated paraconsistent logics between $\Gnsimleq$ and CPL. We also identify a large family of saturated paraconsistent logics in the family of intermediate logics for degree-preserving finite-valued {\L}ukasiewicz logics.
\end{abstract}

\section{Introduction} \label{sect1}

Contradictions frequently arise in scientific theories, as well as in philosophical argumentation. In computer science, techniques for dealing with contradictory information need to be developed in areas such as logic programming, belief revision, the semantic web and artificial intelligence in general. Since classical logic --as well as many other non-classical logics-- trivialize in the presence of inconsistencies, it can be useful to consider logical systems tolerant to  contradictions in order to formalize such situations.

A  logic $L$ is said to be  {\em paraconsistent} with respect to a negation connective $\neg$ when it contains a $\neg$-contradictory but not trivial theory. Assuming that $L$ is (at least) Tarskian, this is equivalent to say that the $\neg$-explosion rule $$\frac{\, \varphi \hspace{4mm} \neg \varphi \,}{\psi}$$ is not valid in $L$. The main challenge for paraconsistent logicians is  defining logic systems in which not only a contradiction does not necessarily trivialize, but also allowing that useful conclusions can be derived from such inconsistent information.

The first systematic study of paraconsistency from the point of view of formal logic is due to da Costa, which introduces in 1963 a hierarchy of paraconsistent systems called $C_n$. This is why da Costa is considered one of the founders of paraconsistency. Under his perspective, propositions in a paraconsistent setting are `dubious' in the sense that, in general, a sentence and its negation can be hold simultaneously without trivialization. That is, it is possible to consider contradictory but nontrivial theories. Moreover,  it is possible to express (in every  system $C_n$) the fact that a given sentence $\varphi$ has a classical behavior w.r.t. the explosion law.
This approach to paraconsistency, in which the explosion law is recovered in a controlled way, was generalized by Carnielli and Marcos in~\cite{car:01} by means of the notion of {\em Logics of Formal Inconsistency} ({\bf LFI}s, in short). An {\bf LFI} is a paraconsistent logic (w.r.t. a negation $\neg$) having, in addition, an unary connective $\circ$ (a {\em consistency operator}), primitive or defined, such that any theory of the form $\{\varphi, \neg\varphi, {\circ}\varphi\}$ is trivial, despite $\{\varphi, \neg\varphi\}$ not being necessarily so. Of course, the main novelty whith respect to da Costa's systems $C_n$ is that the  consistency operator (which corresponds to the well-behavior operator) can now be a primitive connective, which allows to consider a more general and expressive theory of paraconsistency.
{\bf LFI}s have been extensively studied since then (for general references, consult  \cite{car:01,CCM,CC16}). Avron, together with his collaborators, has significantly contributed  to the development of {\bf LFI}s by introducing several new systems, besides the ones proposed in~\cite{car:01,CCM,CC16}, and by providing simple, effective and modular semantics based on non-deterministic matrices (Nmatrices) as well as elegant Gentzen-style proof methods for {\bf LFI}s, see for instance~\cite{ArieliAvron,ArieliAZ10,ArieliAZ11,ArieliAZ11a,AAZbook,Avr:05,Avr:07,Avr:09,Avron17,Avron19,AvronAZ10,avr:lev:01,avr:zam:11}.

According to da Costa, one of the main properties that a paraconsistent logic should have is being as close as possible to classical logic. That is,  a paraconsistent logic should retain as much as possible the classical inferences, and still allowing to have non-trivial, contradictory theories. A natural way to formalize this {\em desideratum} is by means of the notion of maximality of a logic w.r.t. another one. A (Tarskian and structural) logic $L_1$ is said to be {\em maximal} w.r.t. another logic $L_2$ if  both are
defined over the same signature, the consequence relation of $L_1$ is contained in that of $L_2$ (i.e., $L_2$ is an extension of $L_1$) and,  if $\varphi$  is a theorem of $L_2$ which is not derivable in $L_1$,  then the extension of $L_1$ obtained by adding $\varphi$ (and  all of its instances under uniform substitutions) as a theorem coincides with $L_2$. Hence, a `good' paraconsistent logic $L$ should be maximal w.r.t. classical logic CPL (presented over the same signature as $L$). As observed in~\cite{CoEsGiGo}, the notion of maximality can be vacuously satisfied when both logics ($L_1$ and $L_2$) have the same theorems.

In~\cite{ArieliAZ10}, Arieli, Avron and Zamansky propose an interesting notion of maximality w.r.t. paraconsistency: a paraconsistent logic is {\em maximally paraconsistent} if no proper extension of it is paraconsistent. Thus, they prove that
several well-known 3-valued logics such as Sette's P1 and da Costa and D'Ottaviano's  $\mathsf{J}_3$ are maximally paraconsistent. Note that both P1 and  $\mathsf{J}_3$ are also maximal w.r.t. CPL.

These strong features satisfied by logics such as  P1 and  $\mathsf{J}_3$ lead
Arieli, Avron and Zamansky  to introduce  in \cite{ArieliAZ11a} the notion of ideal paraconsistent logics.  Briefly,  a logic $L$ is called {\em ideal  paraconsistent} when it is maximally paraconsistent and maximal w.r.t.  to classical logic CPL (the formal definition of ideal paraconsistency will we recalled in Section~\ref{ideal}). One interesting problem is to find ideal paraconsistent logics, and in this sense~\cite{ArieliAZ11a} provides a vast variety of examples of ideal paraconsistent finite-valued logics, aside from P1 and $\mathsf{J}_3$.

As mentioned above, one of da Costa's requirements for defining reasonable paraconsistent logics is  maximality w.r.t. CPL. Many paraconsistent logicians (probably including Avron and his collaborators) would agree with the relevance of this feature. However, this position is by no means uncontroversial.
In \cite{WansOdin}, Wansing and Odintsov extensively  criticized that requirement. According to these authors, maximality w.r.t.\ classical logic is not a good choice. On the one hand,  the phenomenon of paraconsistency should be interpreted from an  informational perspective instead of considering epistemological or ontological terms. Indeed, the authors claim that ``logic should avoid as many ontological commitments as possible''.\footnote{\cite{WansOdin}, p. 179.}  To this end, they argue that, by definition, logic ``is committed to the existence of languages but not necessarily to the existence of language users''.\footnote{Ibid., p. 180.} This means that, despite the models for logics cannot avoid  linguistic entities, valid inferences should not refer to notions such as `knowledge', `belief states' of any other epistemic or doxastic subjects. Thus, it would be preferable to motivate a system of paraconsistent logic in terms of {\em information},  without appealing to epistemological or ontological commitments such as language users, epistemic subjects possessing mental states, etc.. For instance,  by considering that formulas in an inference process are  pieces of information, the fact that in a paraconsistent logic $\{A, \neg A\}$ does not entail $B$ can be read as `it is just not the case that $\{A, \neg A\}$ provides the information that $B$'. The following are some excerpts from  \cite{WansOdin}:

\begin{quotation}
``classical logic is not at all a natural reference logic for reasoning
about information and information structures. On the other hand, it is reasoning about
information that suggests paraconsistent reasoning.'''\footnote{Ibid., p. 181.}
%\ \ \cite[p. 181]{WansOdin}
\end{quotation}

\begin{quotation}
``one may wonder why exactly a {\em nonclassical} paraconsistent logic, if correct, should have a distinguished
status in virtue of being faithful to classical logic ``as much as possible''.''\footnote{Ibid., p. 181.}
% \ \ \cite[p. 181]{WansOdin}
\end{quotation}

\begin{quotation}
``Paraconsistency does deviate from logical orthodoxy, but it is not at all clear that classical logic indeed is the logical orthodoxy from which paraconsistent logics ought to deviate only minimally.'''\footnote{Ibid., p. 183.}
%\ \ \cite[p. 183]{WansOdin}
\end{quotation}

Although it could be argued against this emphatic perspective, it also seems that being maximal w.r.t. CPL should not be a necessary requirement for being an `ideal' (meaning `optimal') paraconsistent logic.\footnote{It is worth noting that, more recently, the authors have changed the terminology ``ideal paraconsistent logic'' in \cite{ArieliAZ11a} to ``fully maximal and normal paraconsistent logic'' e.g. in \cite{AAZbook}. According to them, they choose the latter ``to use a more neutral terminology'' (see~\cite[Footnote~9, p.~57]{AAZbook}).} This is why we propose in this paper  the notion of {\em saturated paraconsistent} logic, which is just a weakening of the concept of ideal paraconsistent logic, by dropping the requirement of maximality w.r.t.\ CPL. As we shall see along this paper, there are several interesting examples of saturated paraconsistent logics.

While paraconsistency deals with excessive or dubious information,  fuzzy logics were designed for reasoning with imprecise information;
in particular, for reasoning with propositions containing vague predicates.
Given that  both paradigms are able to deal with information -- unreliable, in the case of paraconsistent logics, and imprecise, in the case of fuzzy logics -- it seems reasonable to consider logics which combine both features, namely,  paraconsistent fuzzy logic. The first steps along this way were taken in~\cite{EEFGN}, where a consistency operator $\circ$ was defined in terms of the other connectives (for instance, by using the Monteiro-Baaz $\Delta$-operator) in several fuzzy logics.
In~\cite{CoEsGo:14} this approach was generalized to fuzzy  {\bf LFI}s extending the logic MTL of pre-linear (integral, commutative, bounded) residuated lattices, in which the consistency operator is primitive.

We have studied in different papers paraconsistent logics arising from the family of mathematical fuzzy logics, see e.g. \cite{EEFGN,CoEsGo:14,CoEsGo,CoEsGiGo}.  In particular, in \cite{EEFGN} the authors observe that even though all truth-preserving fuzzy logics $L$ are explosive, their degree-preserving companions $L^{\text{\tiny $\leq$}}$ (as introduced in \cite{BEFGGTV}) are paraconsistent in many cases. This provides a large family of paraconsistent fuzzy logics. In \cite{CoEsGo} the authors studied the lattice of logics between the $n$-valued {\L}ukasiewicz logics \L$_n$ and their degree-preserving companions \L$_n^{\text{\tiny $\leq$}}$. Although there are  many paraconsistent logics for each $n$, no one of them is ideal. However, in \cite{CoEsGiGo} the authors of this paper consider {a wide class of}  logics between  \L$_n^{\text{\tiny $\leq$}}$ and CPL, and in that case they indeed find and axiomatically characterize a family of ideal paraconsistent logics.

In this paper we study paraconsistent logics arising from G\"odel fuzzy logic expanded with an involutive negation G$_\sim$, introduced in  \cite{EsGoHaNa}, as well as from its finite-valued extensions G$_{n\sim}$.
It is well-known that G\"odel logic G coincides with its degree-preserving companion (since
G has the deduction-detachment theorem), but this is not the case for G$_\sim$. In fact, G$_\sim$ and $\Gsimleq$ are different logics, and moreover, while $\Gsimleq$ is explosive w.r.t.\ G\"odel negation $\neg$, it is paraconsistent w.r.t.\ the involutive negation $\sim$.\footnote{In fact, $\Gsimleq$ is then a {\em paradefinite}
logic (w.r.t.\ $\ninv$)  in the sense of Arieli and Avron
\cite{ArieliAvron}, as it is both paraconsistent and paracomplete,
since the law of excluded middle $\varphi \lor \ninv \varphi$ fails,
as in all fuzzy logics.
Logics with a negation which is both paraconsistent and paracomplete
were already considered in the literature under different names: {\em
non-alethic} logics (da Costa) and {\em paranormal} logics (Beziau).
} We also study the logics between $\Gnsimleq$ (the finite valued G\"odel logic with an involutive negation) and CPL, and we find that the ideal paraconsistent logics of this family are only the above mentioned 3-valued logic $\mathsf{J}_3$ and its 4-valued version $\mathsf{J}_4$,
introduced in \cite{CoEsGiGo}.
Moreover, we fully characterize the ideal and the saturated paraconsistent logics between $\Gnsimleq$ and CPL.

The paper is structured as follows. After this introduction, some basic  definitions and known results to be used along the paper will be presented. In Section~\ref{sect-G-tilde} we show that the logics between $\Gsimleq$ and CPL are defined by matrices over a $\rm{G}_\sim$-algebra with lattice filters, and in particular we study the logics defined by matrices over $[0,1]_\sim$ with order filters. In Section~\ref{sect-Gn-tilde} we study the case of finite-valued G\"odel logics with involution $\rm{G}_{n\sim}$, and we observe that $\rm{G}_{3\sim}$ and $\rm{G}_{4\sim}$ coincide with \L$_3$ and \L$_4$ (the 3 and 4-valued {\L}ukasiewicz logics) already studied in \cite{CoEsGiGo}. We prove that, in the general case, each finite $\rm{G}_{n\sim}$-algebra is a direct product of subalgebras of $\bf GV_{n\sim}$, the G\"odel chain of $n$ elements with the unique involution $\ninv$ one can define on it.
This result allows us to characterize the logics between $\Gnsimleq$ and CPL. In Section~\ref{ideal} the definition of saturated paraconsistent logic is formally introduced, and it is proved that between $\Gnsimleq$ and CPL there are only three saturated paraconsistent logics: two of them ($\mathsf{J}_3$ and $\mathsf{J}_4$) are already known and are in fact ideal paraconsistent, and there is only one that is saturated but not ideal paraconsistent, that we call $\mathsf{J}_3 \times \mathsf{J}_4$. Finally, in Section~\ref{sect-luka} we return to the study of finite-valued {\L}ukasiewicz logic and prove that in this framework there is a large family of saturated paraconsistent logics that are not ideal paraconsistent. Some concluding remarks are discussed in the final section.

\section{Preliminaries}

\subsection{Truth-preserving G\"odel logics} \label{truthpreserving}

This section is devoted to needed preliminaries on the G\"odel fuzzy logic G, its axiomatic extensions, as well as their expansions with an involutive negation. We present their syntax and semantics, their main logical properties and the notation we use throughout the article.

The language of G\"odel propositional logic is built as usual from a countable set of propositional variables $V$, the constant $\bot$
and the binary connectives $\land$ and $\to$. Disjunction and negation are respectively  defined as $\varphi \lor \psi := ((\varphi \to \psi) \to \psi) \land  ((\psi \to \varphi) \to \varphi)$ and $\neg \varphi := \varphi \to \bot$, equivalence is defined as $\varphi \leftrightarrow \psi := (\varphi \to \psi) \wedge (\psi \to \varphi)$, and the constant $\top$ is taken as $\bot \to \bot$.

The following are the {\em axioms\/} of $\rm{G}$:\footnote{This axiomatization comes from adding axiom (A7) to the axioms of H\'ajek's BL logic \cite{Hajek98}. Later it was shown that axioms (A2) and (A3) were in fact redundant, see \cite{BeCiHa}  for a detailed exposition and the references therein.} \\

\begin{tabular}{l l}
(A1) &$(\varphi \rightarrow \psi) \rightarrow ((\psi \rightarrow
\chi) \rightarrow (\varphi \rightarrow \chi))$
\\
(A2) &$(\varphi \wedge \psi) \rightarrow \varphi$
\\
(A3) &$(\varphi \wedge \psi) \rightarrow (\psi \wedge \varphi)$
\\
(A4a) &$(\varphi \rightarrow (\psi \rightarrow \chi ))
\to ((\varphi \wedge \psi) \rightarrow \chi)$
\\
(A4b) &$((\varphi \wedge \psi) \rightarrow \chi) \rightarrow
(\varphi\rightarrow (\psi \rightarrow \chi))$
\\
(A5) & $((\varphi \rightarrow \psi) \rightarrow \chi) \rightarrow
 (((\psi \rightarrow \varphi ) \rightarrow
\chi) \rightarrow \chi)$
\\
(A6) & $\bot \rightarrow \varphi$ 
\\
(A7) & $\varphi \to (\varphi \wedge \varphi)$
\\ \\
\end{tabular}
\ \\
The {\em deduction rule\/} of $\rm{G}$ is modus ponens.

As a many-valued logic, G\"odel logic is the axiomatic extension of H\'ajek's Basic Fuzzy Logic BL \cite{Hajek98} (which is the logic of continuous t-norms and their residua)  by means of the contraction axiom (A7).

Since the unique idempotent continuous t-norm is the minimum, this yields that G\"odel logic is strongly complete with respect to its standard fuzzy semantics that interprets formulas over the structure $[0, 1]_\mathrm{G} = ([0, 1], \min, \Rightarrow_\mathrm{G}, 0, 1)$,\footnote{Called {\em standard} G\"odel algebra.} i.e. semantics defined by truth-evaluations of formulas $e$ on $[0, 1]$, where $1$ is the only designated truth-value, such that $e(\varphi \land \psi) = \min(e(\varphi), e(\psi))$,  $e(\varphi \to \psi) = e(\varphi) \Rightarrow_\mathrm{G} e(\psi)$ and $e(\bot) = 0$, where $\Rightarrow_\mathrm{G}$ is the binary operation on $[0, 1]$ defined as
$$ x \Rightarrow_\mathrm{G} y = \left \{
\begin{array}{ll}
1, & \mbox{if } x \leq y \\
y, & \mbox{otherwise}.
\end{array}
\right .
$$
As a consequence, $e(\varphi \lor \psi) = \max(e(\varphi), e(\psi))$ and $e(\neg \varphi) = \neg_\mathrm{G} e(\varphi) = e(\varphi) \Rightarrow_\mathrm{G} 0$. By definition, $\Gamma \models_\mathrm{G} \varphi$ \ iff, for every evaluation $e$ over $[0, 1]_\mathrm{G}$, if $e(\gamma) =1$ for every $\gamma \in \Gamma$  then $e(\varphi)=1$.

G\"odel logic can also be seen as the axiomatic extension of intuitionistic propositional logic  by the prelinearity axiom $$(\varphi \to \psi) \lor (\psi \to \varphi).$$  Its algebraic semantics is therefore given by the variety of prelinear Heyting algebras, also known as G\"odel algebras. A G\"odel algebra is thus a (bounded, integral, commutative) residuated lattice ${\bf A} = (A, \land, \lor, *, \Rightarrow, 0, 1)$ such that the monoidal operation $*$ coincides with the lattice meet $\land$,  
and such that the pre-linearity equation $$(x \Rightarrow y) \lor (y \Rightarrow x) = 1 $$
is satisfied, where $x \lor y  = ((x \Rightarrow y) \Rightarrow y) * ((y \Rightarrow x) \Rightarrow x)$). G\"odel algebras are locally finite, i.e. given a G\"odel algebra $\bf A$ and a finite set $F \subseteq A$, the G\"odel subalgebra generated by $F$ is finite as well.

 It is also well-known that the axiomatic extensions of G\"odel logic correspond to its finite-valued counterparts. If we replace the unit interval $[0, 1]$ by the truth-value set
$GV_n = \{0, 1/(n-1), \ldots, (n-2)/(n-1), 1\}$ in the standard G\"odel algebra $[0, 1]_G$  then the structure ${\bf GV_n} = (GV_n, \min, \Rightarrow_\mathrm{G}, 0, 1)$ becomes the ``standard'' algebra for the $n$-valued G\"odel logic G$_n$, that is the axiomatic extension of G with the axiom \\

 \begin{tabular}{l l}
(A$_{\mathrm{G}_n}$) & $ (\varphi_1 \to \varphi_2) \lor \ldots \lor (\varphi_n \to \varphi_{n+1})$ \\
\end{tabular}
\mbox{} \\

\noindent By definition, $\Gamma \models_{\mathrm{G}_n} \varphi$ \ iff, for every evaluation $e$ over $GV_n$, if $e(\gamma) =1$ for every $\gamma \in \Gamma$ then $e(\varphi)=1$.  In fact the logics G$_n$ are all the axiomatic extensions of G, and for each $n$, G$_n$ is an axiomatic extension of G$_{n+1}$, where G$_2$ coincides with CPL. Thus the set of axiomatic extensions of $\rm{G}$ form a chain of logics (and of the corresponding varieties of algebras) of strictly increasing strength:  $${\rm G} < \ldots < {\rm G}_{n+1} < {\rm G}_n < \ldots < {\rm G}_3 < {\rm G}_2 = {\rm CPL}$$
where $L < L'$ denotes that $L'$ is an axiomatic extension of $L$.

Since the negation in G\"odel logics is a pseudo-complementation and not an involution, in \cite{EsGoHaNa} the authors investigate the residuated fuzzy logics arising from  continuous t-norms without non trivial zero divisors and extended with an involutive negation. In particular, they consider the extension of G\"odel logic $\rm{G}$ with an involutive negation $\sim$, denoted as $\rm{G}_\sim$, and axiomatize it.

The {intended semantics\/}
of the $\ninv$ connective on the real unit interval $[0, 1]$ is
an arbitrary 
order-reversing involution $n:[0,1] \to [0,1]$, i.e.\ satisfying $n(n(x)) = x$ and $n(x) \le n(y)$ whenever $x \ge y$.

It turns out that in $\rm{G}_\sim$, with both
negations, $\neg$ and $\ninv$, the projection Monteiro-Baaz connective $\Delta$ is
definable as 
\begin{center}
$\Delta \varphi := \neg \ninv  \varphi$,
\end{center}
and whose semantics on $[0, 1]$ is given by the mapping $\delta: [0,1] \to [0, 1]$ defined as $\delta(1) = 1$ and $\delta(x) = 0$ for $x < 1$.

Axioms of $\rm{G}_{\ninv}$ are those of $\rm{G}$ plus:\footnote{These are the original axioms from  \cite{EsGoHaNa}, see again \cite{BeCiHa} and the references therein for a shorter axiomatization.} \\

\begin{tabular}{l l l}
$(\ninv 1)$ &$\ninv \ninv \varphi \leftrightarrow \varphi$ &
{\sl (Involution)} \\
$(\ninv 2)$ &$\neg \varphi \to \ninv \varphi$& \\
$(\ninv 3)$ &$\Delta(\varphi \to \psi) \to \Delta(\ninv \psi \to \ninv
\varphi)$
& {\sl (Order Reversing)} \\
$(\Delta 1)$ & $\Delta \varphi \vee \neg \Delta \varphi$ &\\
$(\Delta 2)$ & $\Delta (\varphi \vee \psi) \to (\Delta \varphi \vee \Delta
\psi)$&
\\
$(\Delta 5)$ & $\Delta (\varphi \to \psi) \to (\Delta \varphi \to \Delta
\psi)$&
\\ \\
\end{tabular}
\ \\
and inference rules of $\rm{G}_{\ninv}$ are  {\em modus ponens\/} and {\em necessitation\/} for $\Delta$:
$$\frac{\varphi, \varphi \to \psi}{\psi} \hspace*{1cm} \frac{\varphi}{\Delta \varphi}$$

$\rm{G}_{\ninv}$ is an algebraizable logic, whose equivalent algebraic semantics is the quasivariety of $\rm{G}_{\ninv}$-algebras, defined in the natural way, and generated by the class of its linearly ordered members. Among them, the so-called {\em standard} G$_{\ninv}$-algebra, denoted $[0, 1]_{\rm{G}_\sim}$, is the algebra on the real
interval $[0, 1]$ with G\"{o}del truth functions extended by the involutive negation $\ninv x = 1 - x$. This G$_{\ninv}$-chain generates the whole quasi-variety of $\rm{G}_{\ninv}$-algebras. In fact, we have a strong {\em standard} completeness result for G$_{\ninv}$, see \cite{EsGoHaNa,EsGoMa}: for any set $\Gamma \cup \{\varphi\}$ of $\rm{G}_{\ninv}$-formulas, $\Gamma \vdash_{\rm{G}_{\ninv}} \varphi$ iff $\Gamma \models_{\rm{G}_{\ninv}} \varphi$, where the latter means: for every evaluation $e$ over $[0, 1]_{\rm{G}_\sim}$, if $e(\gamma) =1$ for every $\gamma \in \Gamma$ then $e(\varphi)=1$.

Finally,  let us mention that, while G enjoys the usual deduction-detachment theorem (i.e. $\Gamma \cup\{\varphi\} \vdash_{{\rm G}} \psi$ iff $\Gamma \vdash_{{\rm G}} \varphi \to \psi$),
this is not the case for ${\rm G}_{\ninv}$, that has only the following form of $\Delta$-deduction theorem:
$\Gamma \cup\{\varphi\} \vdash_{{\rm G}_{\ninv}} \psi$ iff $\Gamma \vdash_{{\rm G}_{\ninv}} \Delta\varphi \to \psi$. See also the handbook chapter \cite{EsGoMa} for further properties of ${\rm G}_{\ninv}$.

On the other hand, as in the case of G\"odel logic, one can also consider the logics ${\rm G}_{n\sim}$ for each $n \geq 2$, the finite-valued counterparts of ${\rm G}_\sim$. Namely, ${\rm G}_{n\sim}$ can be obtained as the axiomatic extension of ${\rm G}_\sim$ with the axiom  (A$_{{\rm G}_n}$),\footnote{Equivalently, as the expansion of ${\rm G}_n$ with $\sim$ along with the axioms $(\ninv 1)$-$(\ninv 3)$, $(\Delta 1)$-$(\Delta 3)$, and the necessitation rule for $\Delta$.} and can be shown to be complete  with respect to its intended algebraic semantics, the variety of algebras generated by the linearly ordered algebra  $\bf GV_{n\sim}$ obtained in turn by expanding $\bf GV_n$ with the involutive negation ${\sim} x = 1-x$, the only involutive order-reversing mapping that can be defined on $GV_n$. Thus, for any set $\Gamma \cup \{\varphi\}$ of ${\rm G}_{n\sim}$-formulas, $\Gamma \vdash_{{\rm G}_{n\sim}} \varphi$ iff $\Gamma \models_{{\rm G}_{n\sim}} \varphi$, where the latter means: for every evaluation $e$ over the expansion of $\bf GV_n$ by ${\sim}$, if $e(\gamma) =1$ for every $\gamma \in \Gamma$  then $e(\varphi)=1$. Clearly, ${\rm G}_{2\sim} = {\rm CPL}$. The graph of axiomatic extensions of ${\rm G}_{n\sim}$ is depicted in Fig.~\ref{G4}, where edges denote extensions. It can be observed that, if $n$ is even  then ${\rm G}_{n\sim}$ is an extension of ${\rm G}_{m\sim}$ for any $m > n$, while if $n$ is odd, ${\rm G}_{n\sim}$ is an extension of ${\rm G}_{m\sim}$ only for those $m > n$ being odd as well.  Also, note that, in the figure, ${\rm G}^-_\sim$ denotes the extension of ${\rm G}_\sim$ with the axiom \\

\begin{tabular}{ll}
(NFP) & $\ninv\Delta(\varphi \leftrightarrow \ninv \varphi)$
\end{tabular}
 \mbox{ } \\

 \noindent that captures the condition that the involutive negation does not have a fixpoint, a condition satisfied by all the logics ${\rm G}_{n\sim}$ with $n$ even.

\begin{figure}[h] \label{G_sim}
\begin{center}
\includegraphics[width=7cm]{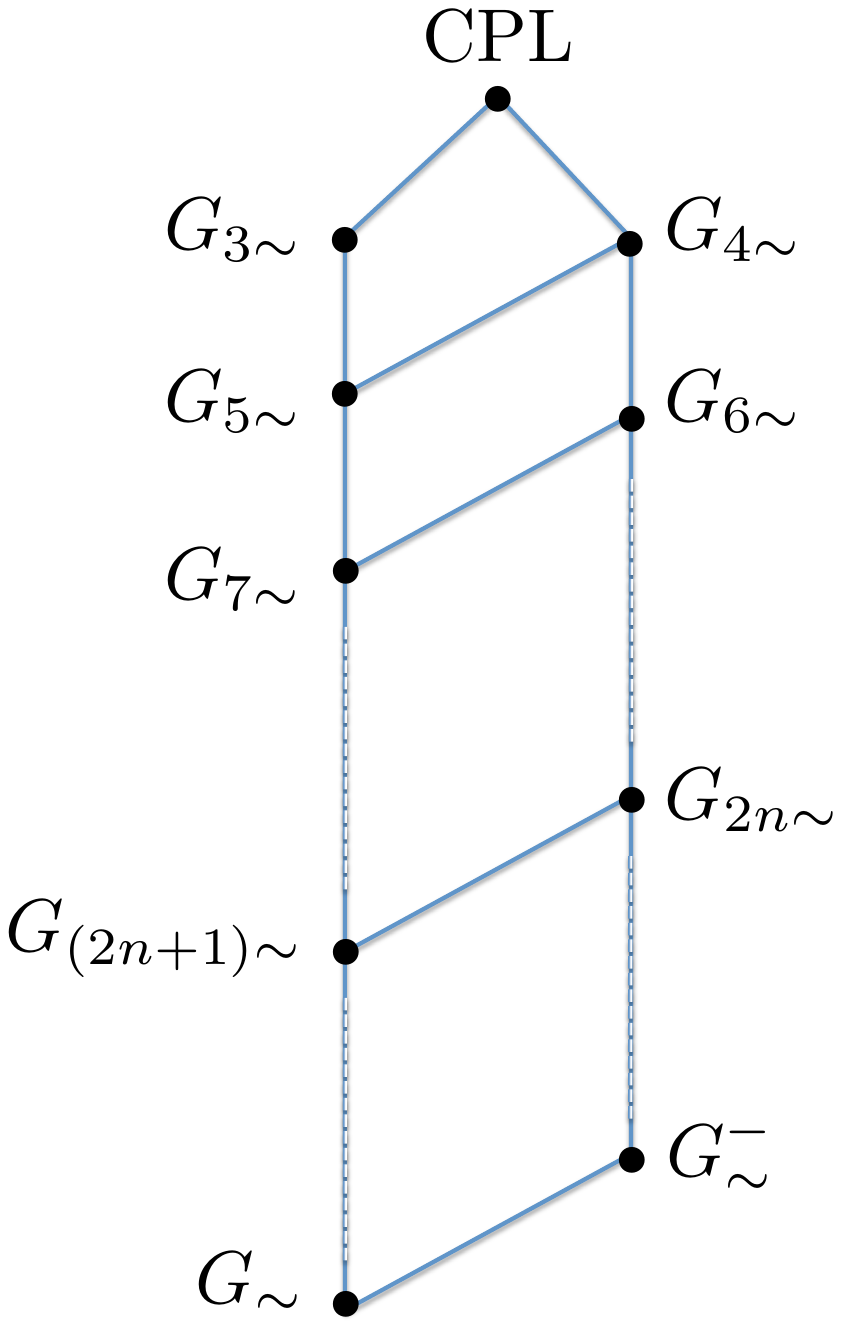}
\caption{Graph of axiomatic extensions of  $\rm{G}_{\sim}$. }
\label{G4}
\end{center}
\end{figure}

\subsection{Degree-preserving G\"odel logics with involution}

Main logics studied in Mathematical Fuzzy Logic are (full) truth-preserving fuzzy logics, like the G\"odel logics introduced in the previous section.  But we can also find in the literature companion logics that preserve degrees of truth, see e.g.~\cite{Font-GTV,BEFGGTV}. {It has been argued  in~\cite{Font} that this approach is more coherent with the commitment of many-valued logics to truth-degree semantics because all values play an equally important role in the corresponding notion of consequence.}
Namely, given a fuzzy logic $\logic{L}$,\footnote{For practical purposes, we can assume in this paper that $\logic{L}$ is an axiomatic extension of H\'ajek's BL logic.  } one can introduce a variant of $\logic{L}$ that is usually denoted  $\logic{L}^{\mbox{\tiny $\leq$ }}$,  whose associated  deducibility relation has the following semantics: for every set of formulas $\Gamma\cup\{\varphi\}$, \\

\noindent \begin{tabular}{lll}
$\Gamma\vdash_{\logic{L}^{\mbox{\tiny $\leq$ }}}\varphi$ &iff & for every $\logic{L}$-chain $\alg{A}$, every $a \in A$, and every $\alg{A}$-evaluation $e$, \\
& & if $a \leq e(\psi)$ for every $\psi \in \Gamma$, then $a \leq e(\varphi)$.
\end{tabular}
\mbox{} \\

\noindent For this reason $\logic{L}^{\mbox{\tiny $\leq$ }}$ is known as a fuzzy logic {\em preserving degrees of truth}, or the {\em degree-preserving companion} of $\logic{L}$. 
It is not difficult to show that $\logic{L}$ and $\logic{L}^{\mbox{\tiny $\leq$ }}$ have the same theorems and also that for every finite set of formulas $\Gamma \cup \{\varphi\}$:
$$ \Gamma \vdash_{\logic{L}^\leq} \varphi \mbox{ iff } \ \vdash_{\logic{L}} \Gamma^\wedge \to \varphi$$
where $\Gamma^\wedge$ means $\gamma_1 \wedge \ldots \wedge \gamma_k$ for $\Gamma = \{\gamma_1,\ldots, \gamma_k\}$ (when $\Gamma$ is empty then $\Gamma^\wedge$ is~$\top$).

{
\begin{remark} It is worth noting that the idea of degree-preserving consequence relations is already present in the context of (classical) modal logic. As it is well-known, under the usual Kripke relational semantics one can consider in modal logic two notions of consequence relation: a {\em local}  and a {\em global} one.\footnote{Given a class of Kripke models,
a formula $\varphi$ follows {\em locally} from a set $\Gamma$ of formulas if, for any Kripke model $M$ in the class and every world $w$ in $M$, $\varphi$ is true in $\langle M,w\rangle$ whenever every formula in $\Gamma$ is true in $\langle M,w\rangle$ as well.  On the other hand,  $\varphi$ follows {\em globally} from $\Gamma$ in the class if, for every Kripke model $M$,  $\varphi$ is true in $\langle M,w\rangle$ for every $w$ whenever every formula in $\Gamma$ is true in $\langle M,w\rangle$ for every $w$.} But modal logics have also been given algebraic semantics by means of the so-called modal algebras. Given such a modal algebra, one can define associated truth-preserving and degree-preserving consequence relations, in an analogous way as done above for a given linearly-ordered L-algebra. It is immediate to see that, for modal logics, local Kripke semantics corresponds to degree-preserving algebraic semantics, while global semantics corresponds to truth-preserving semantics, see e.g. \cite[Defs. 1.35 and 1.37]{modalbook}.
\end{remark}
}

As regards to axiomatization, the logic $\logic{L}^{\mbox{\tiny $\leq$ }}$ admits a Hilbert-style axiomatization having  the same axioms as $\logic{L}$ and the following deduction rules \cite{BEFGGTV}:
\begin{description}
\item[(Adj-$\wedge$)] from $\varphi$ and $\psi$ derive $\varphi\wedge\psi$
\item[(MP-$r$)] if $\vdash_{\logic{L}}\varphi\to\psi$, then from $\varphi$ and $\varphi \to \psi$, derive $\psi$
\end{description}
Note that (MP-$r$) is a restricted form of the Modus Ponens rule, it is only applicable when $\varphi\to\psi$ is a  theorem of $\logic{L}$. 

Since G\"odel logic G enjoys the deduction-detachment theorem, a key observation is that ${ \rm G}^{\mbox{\tiny $\leq$}} = {\rm G}$. However, the case is different for the expansion of G with an involutive negation, since G$_\sim$ does not satisfy the usual deduction-detachment theorem, and hence G$_\sim$ and $\Gsimleq$ are different logics. Moreover, while  $\Gsimleq$ keeps being $\neg$-explosive, it is $\sim$-paraconsistent. Indeed, there are $\varphi, \psi$ such that $\varphi \land \ninv \varphi \not\vdash_{\Gsimleq} 
\psi$. Take, for instance $\varphi$ and $\psi$ as being two different propositional variables, and $e$ a truth-evaluation over $[0, 1]_{\rm{G}_\sim}$ such that $e(\varphi)=\frac{1}{2}$ and $e(\psi) < \frac{1}{2}$.

As for the axiomatization of $\Gsimleq$, we need to consider an extra rule regarding $\Delta$. As shown in \cite{EEFGN}, a complete Hilbert-style axiomatization for  $\Gsimleq$ can be obtained by the axioms of G$_\sim$, the previous rules {\bf (Adj-$\wedge$)} and {\bf (MP-$r$)},\footnote{For $\logic{L} = {\rm G}_\sim$.} together with the following restricted form of the usual necessitation rule for $\Delta$:
\begin{description}
\item[($\Delta$Nec-$r$)] if $\vdash_{{\rm G}_\sim} \varphi$, then from $\varphi$ derive $\Delta \varphi$
\end{description}

Finally, let us consider the logics  $\Gnsimleq$, the degree-preserving companions of the finite-valued logics $\Gsimleq$, defined in the obvious way as above for $L = {\rm G}_{n \sim}$.
Similarly to $\Gsimleq$,  $\Gnsimleq$ also admits the following Hilbert-style axiomatization: $\Gnsimleq$ has as axioms those of ${\rm G}_{n \sim}$, and as rules, the rule {\bf (Adj-$\wedge$)} and  the following restricted rules:

\begin{description}
\item[(MP-$r$)] if $\vdash_{{\rm G}_{n \sim}} \varphi\to\psi$, then from $\varphi$ and $\varphi \to \psi$, derive $\psi$
\item[($\Delta$Nec-$r$)] if $\vdash_{{\rm G}_{n \sim}}\varphi$, then from $\varphi$ derive $\Delta \varphi$
\end{description}

\section{Logics defined by matrices over $[0,1]_{G \sim}$ by means of order filters} \label{sect-G-tilde}

By  a {\em logical matrix} we understand a pair $\langle {\bf A},F \rangle$ where {\bf A} is an algebra and $F$ is a subset of the domain $A$ of {\bf A}.  The logic $L(M)$ 
defined by the matrix $M = \langle {\bf A},F \rangle$ is obtained by stipulating, for any set of formulas $\Gamma \cup \{\varphi\}$, \\

\begin{tabular}{ lll }$ \Gamma \vdash_{L(M)} \varphi$ & if & for every evaluation $e$ on $\bf A$, \\ & & if $e(\gamma) \in F$ for every $\gamma \in \Gamma$, then $e(\varphi) \in F$.
\end{tabular} \ \\

\noindent On the other hand, the logic $L({\cal M})$ determined by a class of matrices $\cal M$ is defined as the intersection of the logics defined by all the matrices in the family. A logic is said to be a {\em matrix logic} if it is of the form $L({\cal M})$ for some class of matrices $\cal M$.\\

\noindent {\bf Notation:} {\em In the rest of the paper, without danger of confusion and for the sake of a lighter notation, we will often  identify a matrix $M$ or a set of matrices $\cal M$ with their corresponding logics $L(M)$ and $L({\cal M})$}.   \\

As proved in \cite{BEFGGTV} for logics of residuated lattices, one can show that $\Gsimleq$,  the degree-preserving companion of G$_\sim$,  is not algebraizable in the sense of Block and Pigozzi and thus it has no algebraic semantics. But it has a semantics via matrices. Indeed,
$\Gsimleq$ is the logic defined by the set of matrices $${\cal M}_{{\rm G}_\sim} = \{\langle {\bf A}, F\rangle : {\bf A} \mbox{ is a } {\rm G}_{\sim}\mbox{-algebra and } F \mbox{ is a lattice filter of } {\bf A}\}.$$
Using similar arguments as in the proof of \cite[Theorem 2.12]{BEFGGTV}, in fact we can also prove that $\Gsimleq$ is complete
with respect to a subset of ${\cal M}_{{\rm G}_\sim}$, namely the set of matrices over the standard G$_\sim$-algebra
$${\cal M}_{[0, 1]} = \{\langle [0,1]_{{\rm G}_\sim}, F\rangle :  F \mbox{ is an order filter of } [0,1]\}.$$

Next, we study the relationships among all the logics defined by matrices from ${\cal M}_{[0, 1]}$, i.e.\ matrices over the algebra $[0,1]_{\rm{G}_\sim}$ by order filters, identifying which ones are paraconsistent.
Actually, the order filters on $[0,1]_{\rm{G}_\sim}$ are the following sets: $F_{[a} = \{x\in [0,1] : x\geq a\}$ for all $a \in (0, 1]$, and $F_{(a} = \{x\in [0,1] : x > a\}$ for all $a \in [0, 1)$. Abusing the notation, we will denote the corresponding logics as
$${\rm G}^{[a}_\sim = \langle [0,1]_{\rm{G}_\sim}, F_{[a}\rangle \mbox{ and } {\rm G}^{(a}_\sim=\langle [0,1]_{\rm{G}_\sim}, F_{(a} \rangle .$$
The consequence relations corresponding to these logics will be respectively denoted by $\vdash_{[a}$ and $\vdash_{(a}$, while $\vdash^f_{[a}$  and $\vdash^f_{(a}$ will denote the finitary companions of $\vdash_{[a}$ and $\vdash_{(a}$, respectively.\footnote{Recall that the finitary companion of a logic $(L, \vdash)$ is given by $(L, \vdash^f)$ where, for every $\Gamma \cup \{\varphi\} \subseteq L$, $\Gamma \vdash^f \varphi$ iff there exists a finite $\Gamma_0 \subseteq \Gamma$ such that  $\Gamma_0 \vdash \varphi$.}
We will write $\vdash_{1}$ and $\vdash^f_{1}$ instead of $\vdash_{[1}$ and $\vdash^f_{[1}$, respectively. 

Some of the logics $\vdash_{[a}$ and $\vdash_{(a}$ are in fact finitary. This is shown in the next lemma. 

\begin{lemma}  The logic $\vdash_1$ is finitary. Moreover, the logics $\vdash_{(1/2}, \vdash_{[1/2}$ and $\vdash_{(0}$ are equivalent, as deductive systems, to $\vdash_1$ and hence they are finitary as well. Therefore all these logics coincide with their finitary companions $\vdash^f_1, \vdash^f_{(1/2}, \vdash^f_{[1/2}$ and $\vdash^f_{(0}$ respectively.
\end{lemma}

\begin{proof}
\begin{itemize}

\item Since ${\rm G}_{\ninv}$ has a finitary axiomatization (see previous Section \ref{truthpreserving}) and it is strongly standard complete w.r.t. to $\vdash_{1}$, then $\vdash_{1}$ is finitary and  coincides with $\vdash^f_{1}$. 

\item We prove that, in fact, $\vdash_{(1/2}, \vdash_{[1/2}$ and $\vdash_{(0}$ are all of them equivalent to $\vdash_1$, in the sense of Blok and Pigozzi \cite{blok:pig:01}. Indeed, for each formula $\varphi$, define the following transformations: 

- $\varphi^{*_1} := (\ninv \varphi \to \varphi) \land \neg \Delta (\varphi \leftrightarrow \ninv \varphi)$, $\varphi^{*_2} := \ninv \varphi \to \varphi$, $\varphi^{*_3} := \neg\neg \varphi$.

Further, if  $\Gamma$ is a set of formulas, define $\Gamma^* := \{ \psi^* \mid \psi \in \Gamma \}$ for $* \in \{*_1, *_2, *_3  \}$. 
It is easy to check that  for any ${\rm G}_{\ninv}$-evaluation $e$, we have: 

- $e(\varphi) > 1/2$ iff $e(\varphi^{*_1}) =1$, 

- $e(\varphi) \geq 1/2$ iff $e(\varphi^{*_2}) =1$, 

- $e(\varphi) > 0$ iff $e(\varphi^{*_3}) =1$.

Then one can check that the following three conditions are satisfied: 

\begin{itemize}
\item[(i)]  The logics $\vdash_{(1/2}, \vdash_{[1/2}$ and $\vdash_{(0}$ can be faithfully interpreted in $\vdash_1$ as the following equivalences hold: 

\mbox{ } $\Gamma \vdash_{(1/2} \varphi$ iff $\Gamma^{*_1} \vdash_{1} \varphi^{*_1}$,  $\Gamma \vdash_{[1/2} \varphi$ iff $\Gamma^{*_2} \vdash_{1} \varphi^{*_2}$, and \\ \mbox{ } $\Gamma \vdash_{(0} \varphi$ iff $\Gamma^{*_3} \vdash_{1} \varphi^{*_3}$.

\item[(ii)] We can also interpret $\vdash_1$ in any of the other consequence relations by using the $\Delta$ operator, indeed, we have: 

\mbox{ } $\Gamma\vdash_1\varphi$ iff $\Gamma^{\Delta}\vdash_{(1/2}\Delta(\varphi)$ iff $\Gamma^{\Delta}\vdash_{[1/2}\Delta(\varphi)$ iff $\Gamma^{\Delta}\vdash_{>0}\Delta(\varphi)$

where $\Gamma^{\Delta} = \{\Delta(\psi)  : \psi\in\Gamma\}$. 

\item[(iii)]  Finally, the following inter-derivabilities show that the $\Delta$ acts as a proper converse transformation of each $*_i$ in the corresponding logic: 

\mbox{ } $ \psi\dashv\vdash _{(1/2}\Delta(\psi^{*_1})$ and   $ \varphi\dashv\vdash_1 (\Delta \varphi)^{*_1}$,

\mbox{ } $ \psi\dashv\vdash _{[1/2}\Delta(\psi^{*_2})$ and   $ \varphi\dashv\vdash_1 (\Delta \varphi)^{*_2}$, 

\mbox{ } $ \psi\dashv\vdash _{(0}\Delta(\psi^{*_3})$ and   $ \varphi\dashv\vdash_1 (\Delta \varphi)^{*_3}$
\end{itemize}

As a consequence, the logics $\vdash_1, \vdash_{(1/2}, \vdash_{[1/2}$ and $\vdash_{(0}$ are equivalent, and since $\vdash_1$ is finitary, so are the other logics as well. 
\end{itemize}
\end{proof}
Therefore, as a consequence of previous lemma, the only cases left open are whether the logics $\vdash_{[a}$ and $\vdash_{(a}$ are finitary for $a \in (0, 1/2) \cup (1/2, 1)$.

Next proposition shows the relationships among the remaining logics defined by matrices over the algebra $[0,1]_{\rm{G}_\sim}$ by order filters.

\begin{proposition}
The logics ${\rm G}^{[a}_\sim = \langle [0,1]_{\rm{G}_\sim}, F_{[a}\rangle$ for $a \in (0, 1]$,  ${\rm G}^{(a}_\sim=\langle [0,1]_{\rm{G}_\sim}, F_{(a} \rangle$ for $a \in [0, 1)$, and their finitary companions, satisfy the following properties:\footnote{In the following we use $p$ and $n$ to denote positive and negative values in $[0, 1]$ with respect to the negation $\ninv x = 1-x$; in other words, $p > 1/2$ and $n < 1/2$.}
\begin{enumerate}
\item[P1.] $\vdash_{[p} \; = \; \vdash_{[p'}$ and $\vdash_{(p} \; = \; \vdash_{(p'}$, for all $p, p' \in (1/2, 1)$.

Moreover, $\vdash_{(p} \; \subseteq \; \vdash_{[p}$ and $\vdash^f_{[p} \; = \; \vdash^f_{(p}$ for all $p \in (1/2, 1)$.

\item[P2.] $\vdash_{[n} \; = \; \vdash_{[n'}$ and $\vdash_{(n} \; = \; \vdash_{(n'}$, for all $n, n' \in (0, 1/2)$.

Moreover, $\vdash_{(n} \; \subseteq \; \vdash_{[n}$ and $\vdash^f_{[n} \; = \; \vdash^f_{(n}$ for all $n \in (0, 1/2)$.

\item[P3.] $\vdash_{[p} ~\varsubsetneq ~\vdash_1$,  for any $p \in (1/2,1)$.

\item[P4.] $\vdash_1$ and ~$\vdash_{[1/2}$ are not comparable. 
\item[P5.] $\vdash_{[p}$ and ~$\vdash_{[1/2}$, as well as $\vdash^f_{[p}$ and ~$\vdash_{[1/2}$, are not comparable, for any $p \in (1/2,1)$. 
\item[P6.] $\vdash^f_{[p}  ~\varsubsetneq  ~\vdash_{(1/2}$,  for any $p \in (1/2,1)$.
\item[P7.] $\vdash_{[p}$ and ~$\vdash_{[n}$ are not comparable, for any $p \in (1/2,1)$ and any $n \in (0,1/2)$. The same holds for $\vdash^f_{[p}$ and ~$\vdash^f_{[n}$. 
\item[P8.]  $\vdash_{[n} ~\varsubsetneq ~\vdash_{[1/2}$,  for any $n \in (0,1/2)$. 
\item[P9.] $\vdash_{(0}$, $\vdash_{[1/2}$ and $\vdash_{(1/2}$ are not pair-wise comparable. 
\item[P10.]  $\vdash^f_{[n}  \; \varsubsetneq  \; \vdash_{(0}$, for any $n \in (0, 1/2)$.

\end{enumerate}
\end{proposition}

\begin{proof}  \mbox{}

\begin{enumerate}
\item[P1.]  We divide the proof in four steps:

(i) That $\vdash_{[p} \; = \; \vdash_{[p'}$ and $\vdash_{(p} \; = \; \vdash_{(p'}$ is an easy consequence of the fact that for every $p, p' \in (1/2,1)$  it is possible to define an automorphism $f$ of $[0,1]_{G\sim}$ such that $f(p) = p'$. Let us then show that  $\vdash^f_{[p} \; =\;  \vdash^f_{(p}$ for every $p \in (1/2, 1)$.

(ii) Assume $\{\varphi_i : i \in I\} \vdash^f_{[p} \psi$, with $I$ finite, for some $p \in (1/2, 1)$. Let $q$ such that $1/2 < q < p$, and let $e$ be an evaluation such that $e(\varphi_i) > q $ for all $i \in I$. 
Let $p' = \min_{i \in I} e(\varphi_i)$. Obviously $p' > q$. Then, by (i), we also have $\{\varphi_i : i \in I\} \vdash^f_{[p'} \psi$, and therefore we have $e(\psi) \geq p' > q$, and hence $\{\varphi_i :  i \in I\} \vdash^f_{(q} \psi$. Therefore, we have $\vdash^f_{[p}  \; \subseteq  \; \vdash^f_{(q}$ for all $1/2 < q < p$.

(iii) Recall from (i) that $\Gamma \vdash_{(p} \varphi$ iff $\Gamma \vdash_{(p'} \varphi$ for all $1/2 < p' < 1$. Let $p_1, p_2, \ldots, p_n, \ldots $ be an increasing sequence of values $p_i \in (1/2, p)$ such that $\lim_n p_n = p$. 
Suppose $\Gamma \vdash_{(p} \varphi$, and further assume $e(\psi) \geq p$ for all $\psi \in \Gamma$. Clearly, for each $p_i$, $e(\psi) > p_i$ for all $\psi \in \Gamma$. Since $\Gamma \vdash_{(p_i} \varphi$ for each $p_i$, we have that $e(\varphi) > p_i$ for each $p_i$.  Hence $e(\varphi) \geq p$. 

(iv) Assume $\{\varphi_i :  i \in I\} \vdash^f_{(q} \psi$, with $I$ finite, for some $q \in (1/2, 1)$. Let $p$ be such that $ q < p < 1$, and let $e$ be an evaluation such that $e(\varphi_i) \geq p $ for all $i \in I$. 
Let $q' = \min_{i \in I} e(\varphi_i)$. Obviously $q' \geq p$. Then, by (i), we also have $\{\varphi_i :  i \in I\} \vdash^f_{(q'} \psi$, and therefore we have $e(\psi) \geq q' \geq p$, and hence $\{\varphi_i :  i \in I\} \vdash^f_{p} \psi$. Therefore, we have $\vdash^f_{(q}  \; \subseteq  \; \vdash^f_{[p}$ for all $1/2 < q < p$.

\item[P2.] The proofs are analogous to those of P1.

\item[P3.] Assume $\{\varphi_i :  i \in I\} \vdash_{[p} \psi$ for a given $p  \in (1/2,1)$, and let $e$ be an evaluation such that $e(\varphi_i) = 1$ for all $i \in I$. Since it is also true that $e(\varphi_i) \geq p'$ for all $p'  \in (1/2,1)$, by P1 it follows that $\{\varphi_i :  i \in I\} \vdash_{[p'} \psi$ for all $p'  \in (1/2,1)$, and hence $e(\psi) \geq p'$ for all $p'  \in (1/2,1)$, and thus $e(\psi) = 1$. Therefore $\{\varphi_i :  i \in I\} \vdash_1 \psi$.

The strict inclusion can be easily noticed since, e.g. it holds that  $\varphi \vdash_1 \Delta \varphi$ but $\varphi \nvdash_{[p} \Delta \varphi$ for any $p < 1$.

\item[P4.] It clearly holds that, on the one hand, $\Delta(\varphi \leftrightarrow \ninv \varphi) \vdash_{[1/2} \varphi$ but $\Delta(\varphi \leftrightarrow \ninv \varphi) \nvdash_{1} \varphi$, while on the other hand,
$\varphi \vdash_1 \Delta \varphi$ but $\varphi \nvdash_{[1/2} \Delta \varphi$

\item[P5.]  It follows from noticing that $\Delta(\varphi \leftrightarrow \ninv \varphi) \land \varphi \vdash_{[p} \bot$ and $\Delta(\varphi \leftrightarrow \ninv \varphi) \land \varphi \nvdash_{[1/2} \bot$, while
$\Delta(\varphi \leftrightarrow \ninv \varphi) \vdash_{[1/2} \varphi$ and $\Delta(\varphi \leftrightarrow \ninv \varphi) \nvdash_{[p} \varphi$.

\item[P6.]
Assume that, for a given $p  \in (1/2,1)$, $\{\varphi_i :  i \in I\} \vdash_{[p} \psi$, with $I$ finite, and let $e$ be an evaluation such that $e(\varphi_i) > 1/2 $ for all $i \in I$. 
Let $p' = \min_{i \in I} e(\varphi_i)$. Obviously $p' > 1/2$. Then, from P1 we also have $\{\varphi_i :  i \in I\} \vdash_{[p'} \psi$, and therefore we have $e(\psi) \geq p' > 1/2$, and hence $\{\varphi_i :  i \in I\} \vdash_{(1/2} \psi$. Therefore, we have $\vdash_{[p}  \; \subseteq  \; \vdash_{(1/2}$.

That the inclusion is strict follows from observing that $\ninv \Delta (\varphi \to \ninv \varphi) \vdash_{(1/2} \varphi$, but $ \ninv \Delta (\varphi \to \ninv \varphi) \nvdash_{[p} \varphi$.

\item[P7.]  It follows from observing (i) $\Delta(\varphi \leftrightarrow \ninv \varphi) \vdash_{[n} \varphi$ and $\Delta(\varphi \leftrightarrow \ninv \varphi) \nvdash_{[p} \varphi$, and  (ii)
$ \varphi \vdash_{[p} \ninv \Delta (\varphi \to  \ninv \varphi)$ and $\varphi \nvdash_{[n} \ninv \Delta (\varphi \to \ninv \varphi)$.

\item[P8.]  That $\vdash_{[n} ~\subseteq ~\vdash_{[1/2}$  is proved in a similar way to P3. The strict inclusion is a consequence of the following facts:

(i) $\varphi \land \ninv \varphi \vdash_{[1/2} \Delta(\varphi \leftrightarrow \ninv \varphi)$

(ii) $\varphi \land \ninv \varphi \not\vdash_{[n} \Delta(\varphi \leftrightarrow \ninv \varphi)$

Notice that $e(\varphi \land \ninv \varphi) \geq 1/2$ iff $e(\varphi) = 1/2$ iff $e(\varphi \leftrightarrow \ninv \varphi) = 1$, while if $e(\varphi) = n$, then $e(\varphi \land \ninv \varphi) = e(\varphi \leftrightarrow \ninv \varphi) = n$, but $e(\Delta(\varphi \leftrightarrow \ninv \varphi)) = 0$.

\item[P9.]  That $\vdash_{[1/2}$ and $\vdash_{(1/2}$ are not comparable results from noticing e.g. (i) $\Delta(\varphi \leftrightarrow \ninv \varphi) \vdash_{[1/2} \varphi$ but $\Delta(\varphi \leftrightarrow \ninv \varphi) \nvdash_{(1/2} \varphi$, and (ii)
$\varphi \vdash_{(1/2} \ninv \Delta(\varphi \to \ninv \varphi)$ but $\varphi \nvdash_{[1/2} \ninv \Delta(\varphi \to \ninv \varphi)$.

On the other hand, it is easy to check that $\bot$ follows from $\varphi \land \Delta(\varphi \to \ninv \varphi) \land \neg \Delta (\ninv \varphi \to \varphi)$ in $\vdash_{(1/2}$ and $\vdash_{[1/2}$, but not in $\vdash_{(0}$. Conversely, 
$\neg\neg \varphi \land \neg \Delta \varphi  
\vdash_{(0} \varphi \land \ninv \varphi$, but this is neither the case for $\vdash_{(1/2}$ nor for $\vdash_{[1/2}$.

\item[P10.]  Assume $\{\varphi_i :  i \in I\} \vdash_{[n} \psi$ for a given $n  \in (0, 1/2)$ and a finite set $I$, and let $e$ be an evaluation such that $e(\varphi_i) > 0 $ for all $i \in I$. 
Let $n' = \min_{i \in I} e(\varphi_i)$. Obviously $n' > 0$ and $e(\varphi_i) \geq n'$, for all $i \in I$. Then, from P1, we also have that $\{\varphi_i :  i \in I\} \vdash_{[n'} \psi$, and hence we have $e(\psi) \geq n' > 0$.  This means $\{\varphi_i :  i \in I\} \vdash_{(0} \psi$. Therefore, we have $\vdash^f_{[n}  \; \subseteq  \; \vdash_{(0}$.

On the other hand, $\neg \neg\varphi \vdash_{(0}  \varphi$ but  $\neg \neg\varphi \nvdash_{[n}  \varphi$, hence we have proved that $\vdash^f_{[n}  \; \varsubsetneq  \; \vdash_{(0}$.
\end{enumerate}
\end{proof}

A graphical representation of the different logics (consequence relations) involved in the above proposition can be seen in Figure \ref{Gpn}. 
\begin{figure}[h] 
\begin{center}
\includegraphics[width=10cm]{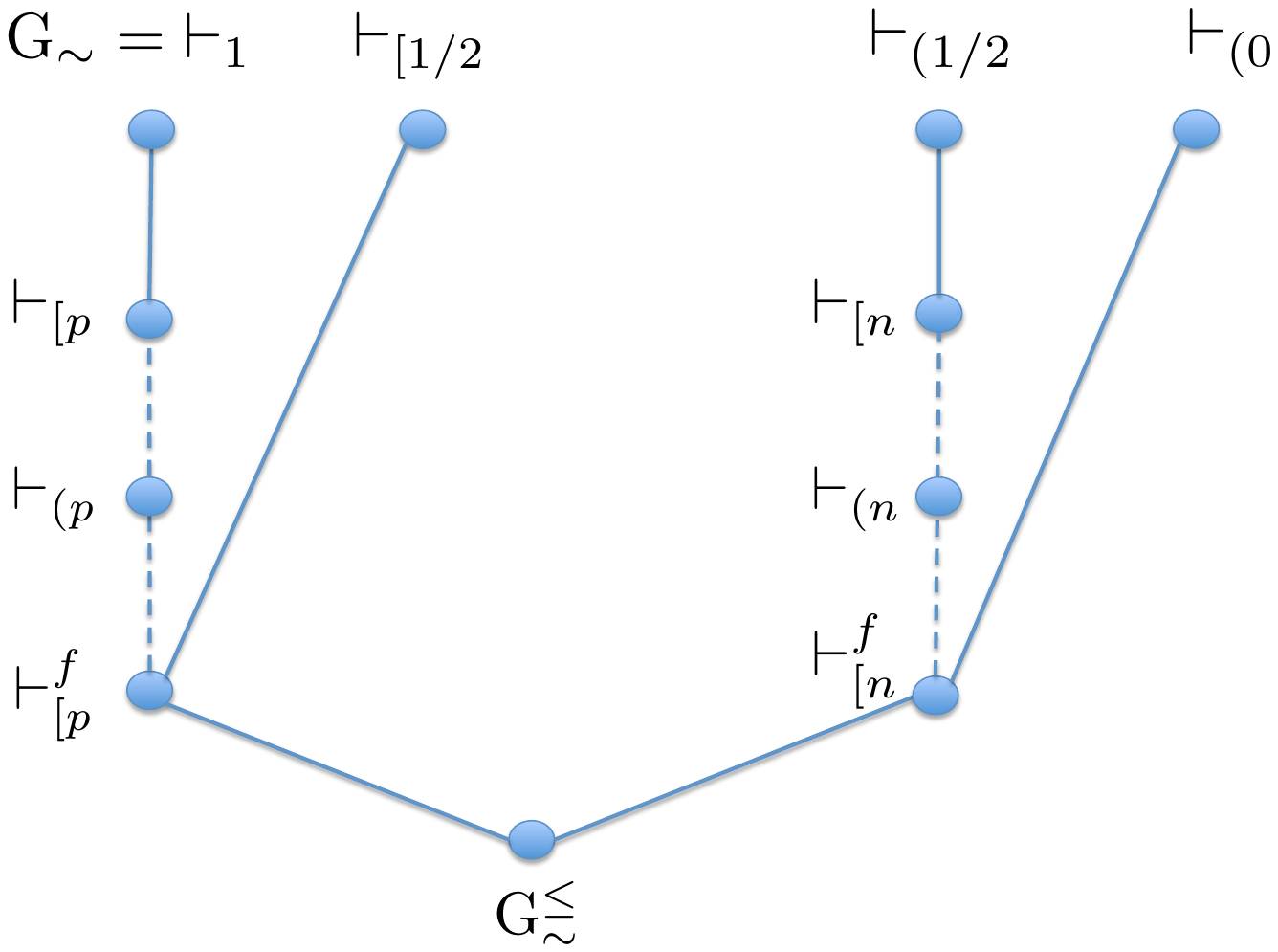}
\caption{Graph of logics over $[0, 1]_{\rm{G}_{\ninv}}$ defined by order filters, where $1/2 < p < 1$ and $0 < n < 1/2$, where edges stand for inclusions (upward sense). The dashed edges denote that it is an open problem whether the connected logics are different.}
\label{Gpn}
\end{center}
\end{figure}

It is clear that a matrix  logic ${\rm G}^{[a}_\sim = \langle [0,1]_{\rm{G}_\sim}, F_{[a}\rangle$ (resp. ${\rm G}^{(a}_\sim=\langle [0,1]_{\rm{G}_\sim}, F_{(a} \rangle$) is paraconsistent only in the case that $a \leq 1/2$ (resp. $a < 1/2$). As a consequence of the above classification, it turns out that there are only three different paraconsistent logics among them.

\begin{corollary} Among the families of logics 
$\{{\rm G}^{[a}_\sim \}_{a \in (0, 1]}$ and $\{{\rm G}^{(a}_\sim\}_{a \in [0, 1)}$: 
\begin{itemize}
\item there are only three different paraconsistent logics:  ${\rm G}^{[a}_\sim$ for any $a \in (0, 1/2)$,  ${\rm G}^{[1/2}_\sim$, and ${\rm G}^{(0}_\sim$.
\item there are only three different explosive logics: ${\rm G}^{[a}_\sim$ for any $a \in (1/2, 1)$,  ${\rm G}^{(1/2}_\sim$, and ${\rm G}^{[1}_\sim$.
\end{itemize}
\end{corollary}

In analogy to \cite[Theorem 2]{CoEsGo}, it is easy to show that every intermediate logic $L$ between $\Gsimleq$ and CPL  is in fact the logic $L({\cal M}')$ defined by a subfamily of matrices ${\cal M}' \subseteq {\cal M}_{{\rm G}_\sim}$.
However, note that the set of G$_{\sim}$-algebras and their lattice filters is very large. Then, an exhaustive analysis of the set of intermediate logics between $\Gsimleq$ and CPL actually seems to be a difficult task. Because of this, in the next section we will restrict ourselves to the case of finite-valued G\"odel logics with an involutive negation G$_{n \sim}$.

\begin{remark} In the last corollary we have shown that ${\rm G}^{(0}_\sim$,  the matrix logic defined by the standard G$_{\ninv}$-algebra $[0,1]_{\rm{G}_\sim}$ and the filter $(0,1]$  of designated values, is a paraconsistent logic. In~\cite{Avron16} {Avron introduces} a paraconsistent extension of the logic {T} of Anderson and Belnap called { FT}. This logic, which intends to be ``a paraconsistent counterpart of \L ukasiewicz Logic $\L_\infty$'' (\cite[pp.~75]{Avron16}), is firstly defined axiomatically over {a propositional language with} connectives $\land$, $\lor$, ${\sim}$, $\to_\mathrm{FT}$,\footnote{In~\cite{Avron16} the symbols $\neg$ and $\to$ were used instead of ${\sim}$ and $\to_\mathrm{FT}$. We adopt this notation in order to keep the notation of the present paper uniform.} and then it {is} proved that {FT} is semantically characterized by the logic matrix defined by the ordered algebra ${\bf M_{[0,1]}} = ( [0, 1], \land, \lor,  {\sim}, \to_\mathrm{FT},0,1)$ and the filter $(0,1]$  of designated values ({this is why Avron considers  {FT} as a logic that preserves {\em non-falsity}}). Here 
%$\leq$ is the usual order in $[0,1]$, 
$\land$, $\lor$ and $\sim$ are defined as in ${[0,1]_{\rm{G}_\sim}}$, while $\to_\mathrm{FT}$ is defined as follows: 
$$x \to_\mathrm{FT} y = \left \{
\begin{array}{l l}
\max(1-x, y), & \mbox{if } x \leq y \\
0, & {\rm otherwise.} 
\end{array} \right .
$$ 
Now, observe that the implication $\to_\mathrm{FT}$ of  ${\bf M_{[0,1]}}$ is definable in  $[0,1]_{\rm{G}_\sim}$ as $x \to_\mathrm{FT} y = \Delta(x \to_{\rm{G}} y) \land (\ninv x \lor y)$. As a consequence of this, the logic {FT} is interpretable in ${\rm G}^{(0}_\sim$ by means of a mapping $*: Fm_\mathrm{FT} \to  Fm_{{\rm{G}}_\sim}$ defined recursively as follows: $p^*=p$ if $p$ is a propositional variable;  $(\ninv \varphi)^* = \ninv \varphi^*$;  $(\varphi\land \psi)^* = \varphi^* \land \psi^*$;  $(\varphi\lor \psi)^* = \varphi^* \lor \psi^*$;  and $(\varphi \to_\mathrm{FT} \psi)^* = \Delta(\varphi^* \to_{\rm{G}} \psi^*) \land (\ninv \varphi^* \lor \psi^*)$. Then, for every $\Gamma \cup \{\varphi\} \subseteq  Fm_\mathrm{FT}$, we have:
$$\Gamma \vdash_\mathrm{FT} \varphi \ \mbox{ iff } \ \Gamma^* \vdash_{(0} \varphi^*,$$
where $\Gamma^*$ denotes the set $\{\psi^* \ : \ \psi \in \Gamma\}$ and $ \vdash_{(0}$ denotes the consequence relation of ${\rm G}^{(0}_\sim$.
Moreover, in~\cite[Example~3.4]{Avron16} Avron considers, for every $n > 1$,  the finite subalgebra $\bf M_{[0,1]}^n$ of $\bf M_{[0,1]}$ with domain $GV_n = \{0, 1/(n-1), \ldots, (n-2)/(n-1), 1\}$. Let $\rm {FT}^n$ be the logic characterized by the logic matrix defined by the algebra $\bf M_{[0,1]}^n$ and the filter $F_{\frac{1}{n-1}} = \{ a \in GV_n \ : \ a > 0\} = GV_n\cap  (0,1]$  of designated values. Then, the interpretation $*$ above also shows that $\rm {FT}^n$ is interpretable in $\langle {\bf GV_{n\sim}},F_{\frac{1}{n-1}}\rangle$, since we also have:
$$\Gamma \vdash_{\rm{FT^n}} \varphi \ \mbox{ iff } \ \Gamma^* \vdash_{\{\frac{1}{n-1}\}} \varphi^*,$$ 
where $\vdash_{\{\frac{1}{n-1}\}}$ is the consequence relation of the matrix logic  $\langle {\bf GV_{n\sim}},F_{\frac{1}{n-1}}\rangle$. This notation will we also used in Subsection~\ref{Gn}. 

As a matter of fact, it can be observed that the $\Delta$ operator of  $\bf [0,1]_{\rm{G}_\sim}$ is definable in $\bf M_{[0,1]}$ as $\Delta x = 1 \to_\mathrm{FT} x$ and so the G\"odel implication $\to_{\rm{G}}$ of  $\bf [0,1]_{\rm{G}_\sim}$  is also definable in $\bf M_{[0,1]}$ as   $x \to_{\rm{G}} y = {\sim} \Delta {\sim}(x \to_\mathrm{FT} y) \lor y$. Observe, however, that the logic FT has neither bottom nor top,\footnote{Indeed, every formula $\varphi \in  Fm_{\rm FT}$ gets the value $1/2$ in any evaluation $e$ over $\bf M_{[0,1]}$ such that $e$ assigns the value $1/2$ to any propositional variable occurring in $\varphi$.} hence there is no formula in FT which can express the $\Delta$ operator. Let FT$_0$ be the logic defined by the same matrix $\langle {\bf M_{[0,1]}},(0,1]\rangle$ of FT, but now over an expanded language $Fm_{\rm FT_0}$ containing a constant $\bot$ and adding the requirement that $e(\bot) = 0$ for every evaluation $e$. Let us denote by $\vdash_{\mathrm{FT}_0}$   its corresponding consequence relation. Consider the mapping $\#: Fm_{\rm{G}_\sim} \to Fm_{\rm FT_0}$ defined recursively as follows: $p^{\#}=p$ if $p$ is a propositional variable; $\bot^{\#}=\bot$;  $(\ninv \varphi)^{\#} = \ninv\varphi^{\#}$; $(\varphi\land \psi)^{\#} = \varphi^{\#} \land \psi^{\#}$;  and  $(\varphi \to_G \psi)^{\#} = {\sim} \Delta {\sim}(\varphi^{\#} \to_\mathrm{FT} \psi^{\#}) \lor \psi^\#$ (where $\Delta \alpha = \ninv \bot \to_\mathrm{FT} \alpha$ for every $\alpha$). Then, $\Gamma \vdash_{(0} \varphi  \mbox{ iff } \Gamma^\# \vdash_{\mathrm{FT}_0} \varphi^\#$, showing that  the logic  ${\rm G}^{(0}_\sim$ is interpretable in FT$_0$. 
Since $\varphi$ is equivalent to $\varphi^{*\#}$ in FT$_0$ for every $\varphi \in Fm_{\rm FT_0}$,  and $\psi$ is equivalent to $\psi^{\#*}$ in ${\rm G}^{(0}_\sim$ for every $\psi \in Fm_{\rm{G}_\sim}$ (here, $*$ is extended to $Fm_{\rm FT_0}$ by putting $\bot^*=\bot$), the logics  ${\rm G}^{(0}_\sim$ and FT$_0$ are the same up to language. 
The same relationship holds between  $\langle {\bf GV_{n\sim}},F_{\frac{1}{n-1}}\rangle$ and the logic $\rm FT_0^n$ obtained from  $\rm FT^n$ by adding $\bot$.

\end{remark}

\section{Logics between $\mathrm{G}\mystrutsec^{\text{\footnotesize $\leq$}}_{n\sim}$ and CPL}
\label{sect-Gn-tilde}

In this section we will study the intermediate logics between $\Gnsimleq$ and CPL, for a natural $n > 2$.
The cases $n=3$ and $n= 4$ are easy to analyze since G$_{3 \sim}$ and G$_{4 \sim}$ coincide respectively with the 3-valued and 4-valued {\L}ukasiewicz logics \L$_{3}$ and \L$_{4}$.

\begin{proposition} \label{3i4}
${\rm G}_{3 \sim}$ and   ${\rm G}_{4 \sim}$ are logically equivalent to \L$_{3}$ and \L$_{4}$ respectively.
\end{proposition}
\begin{proof} The proof is algebraic, we prove that the standard algebras $\bf GV_{3 \sim}$ and $\bf GV_{4 \sim}$ are termwise equivalent to the standard {\L}ukasiewicz algebras of \L$_{3}$ and \L$_{4}$ respectively.
First, in the algebra $\bf GV_{3 \sim}$ it is possible to define the binary connective $x \to_{3\textrm{\L}} y = (x \to y) \lor (\ninv x \lor y)$, that coincides with the 3-valued {\L}ukasiewicz implication, i.e.\ we have $x \to_{3\textrm{\L}} y = \min(1, 1-x+y)$ for every $x, y \in GV_3$. Thus in $\bf GV_{3 \sim}$ we can define all the {\L}ukasiewicz connectives, in other words, $(GV_{3 \sim}, \to_{3\textrm{\L}}, \ninv, 0, 1)$ is in fact the

Second, also in the algebra $\bf GV_{4 \sim}$ we can define the binary connective
$$x  \to_{4\textrm{\L}} y = \ninv x \lor [\Delta(\ninv x \to x) \land (\ninv \Delta x) \land (\neg \neg y) \land x)] \lor (x \to y)$$
which coincides again with the 4-valued {\L}ukasiewicz implication, i.e. $x  \to_{4\textrm{\L}} y = \min(1, 1-x+y)$ for every $x, y \in GV_4$.

On the other hand, in any finite MV-algebra $\bf {\L}V_n$
we can always define G\"odel implication as
$x \to_{\rm{G}} y = \Delta(x \to_{\textrm{\L}} y) \lor y$ and G\"odel negation as $\neg_{\rm{G}} x = \Delta (\ninv x)$.\footnote{Recall that the $\Delta$ connective is definable in any algebra $\bf \textbf{\L} V_n$.}
\end{proof}

\begin{remark} \label{J34}
From the last result, it follows that the logics between $\Gasimleq$ (resp. $\Gbsimleq$) and  CPL coincide with the logics between $\L_{3 }^\leq$ (resp. $\L_{4 }^\leq)$ and  CPL studied in \cite{CoEsGo} and \cite{CoEsGiGo}. Among them, the well-known da Costa and D'Ottaviano's 3-valued logic $\mathsf{J}_3$, that is equivalent (up to language) to the matrix logic $\langle {\bf {\L}V_3}, \{1/2, 1\}\rangle$, is a logic between  $\L_{3 }^\leq$ and  CPL. Analogously, its 4-valued generalization  $\mathsf{J}_4$, defined as the matrix logic $\langle {\bf {\L}V_4}, \{1/3, 2/3, 1\}\rangle$ in \cite{CoEsGiGo}, is a logic between $\L_{4 }^\leq$  and CPL. Therefore, we can consider the logics $\mathsf{J}_3$ and $\mathsf{J}_4$ to be equivalent as well to the intermediate logics $\langle {\bf \bf GV_{3 \sim}}, \{1/2, 1\}\rangle$ and $\langle {\bf \bf GV_{4 \sim}}, \{1/3, 2/3, 1\}\rangle$ respectively.
\end{remark}

Observe, however, that for any $n > 4$, G$_{n \sim}$ is no longer equivalent to $\L_n$. Thus we need to study the intermediate logics for $\Gnsimleq$ for $n > 4$,  and this is the goal of the next subsection, while in Section \ref{n5} we will have a closer look to the case $n = 5$.

\subsection{The intermediate logics of $\mathrm{G}\mystrutsubsec^{\text{\footnotesize $\leq$}}_{n\sim}$  for $n> 4$} \label{Gn}

Throughout this section $n$ will denote a natural number such that $n > 4$.

Following the same arguments as in previous sections, it is easy to check that $\Gnsimleq$ is in fact the logic semantically defined by the class of matrices: 
$$\{\langle {\bf A}, F\rangle  : {\bf A} \mbox{ is a } {\rm G}_{n \sim}\mbox{-algebra and } F \mbox{ is a lattice filter of }{\bf A}\}.$$
Therefore, in order to study the intermediate logics between  $\Gnsimleq$ and CPL we need to characterize the (finite) G$_{n\sim}$-algebras.

\begin{proposition}
Every finite ${\rm G}_{n\sim}$-algebra is a finite direct product of finite ${\rm G}_{n\sim}$-chains.
\end{proposition}
\begin{proof}Notice that for every ${\rm G}_{n \sim}$-chain the term $t(x,y,z):= (\Delta(x\leftrightarrow y)\land z)\lor (\lnot\Delta(x\leftrightarrow y)\land x)$ is a discriminator term,\footnote{In fact, this is a discriminator term in the whole variety of ${\rm G}_{\sim}$-algebras. {For a definition of discriminator term and discriminator variety see \cite{bursan81}.}} hence every ${\rm G}_{n \sim}$-variety  is a discriminator variety. Then the result is a consequence of a result of universal algebra (see for instance \cite[Theorem 9.4, item (d)]{bursan81}).
\end{proof}

In the following we will need to consider products of logical matrices.

\begin{definition} \label{prodmat}
Let $L_i=\langle {\bf A}_i, D_i \rangle$ (for $i \in I$) be a family of logical matrices, where each $D_i$ is an order filter in ${\bf A_i}$. The product of these matrices
  is the logical matrix $L=\Pi_{i \in I} L_i=\langle \Pi_{i \in I}{\bf A}_i, \Pi_{i \in I}D_i \rangle$, where $\Gamma \vdash_L \varphi$ iff, for every tuple of evaluations $(e_i)_{i \in I}$, each $e_i$ over ${\bf A_i}$, the following condition holds: if $e_i(\psi) \in D_i$ for every $i \in I$ and every $\psi \in \Gamma$, then $e_i(\varphi) \in D_i$ for every $i \in I$.
\end{definition}

\begin{remark}
Obviously, a matrix logic $L$ as above is paraconsistent iff all the components $L_i$ are paraconsistent. For example, if one component is $\langle {\bf GV_{2 \sim}}, F_1\rangle$, then the matrix logic is not paraconsistent.
\end{remark}

Since every G$_{\sim}$-algebra is locally finite, every intermediate logic $L$ between $\Gnsimleq$ and CPL  is induced by a family  of product matrices $\langle {\bf A},F\rangle$ where $\bf A$ is a finite direct product of subalgebras of $\bf GV_{n \sim}$ and $F$ is a lattice filter of $A$ compatible\footnote{A filter $F$ of an algebra $\bf A$ is {\em compatible} with a logic $L$ if, whenever $\Gamma \vdash_L \varphi$, the following holds: for every  $\bf A$-evaluation $e$, if $e(\gamma) \in F$ for every $\gamma \in \Gamma$ then $e(\varphi) \in F$.} with $L$.

In \cite{CoEsGo} and \cite{CoEsGiGo} products of logical  matrices were considered for {\L}ukasiewicz finite-valued logics. For instance, \cite{CoEsGo} contains a full description of the set $Int_\Pi(\L_n)$ of logics defined by (sets of) products of matrices over the standard $\L_n$-algebra. Additionally, it also contains an almost full description of the set $Int(\L_n)$ of logics defined by sets of products of matrices over subalgebras of the standard $\L_n$-algebra which are sublogics of $\L_n$,  when $n-1$ is a prime number.

In the rest of this section, we will consider families of intermediate logics  between  $\Gnsimleq$ and CPL of increasing generality. 

First of all we study the matrix logics $\langle {\bf GV_n},F\rangle$ where $F$ is an order filter of $\bf GV_n$. 
In order to simplify the notation, for every nonempty subset $T \subseteq GV_n$ we denote by $L({\cal M}_T )$ the logic defined by the set of matrices ${\cal M}_T = \{\langle {\bf GV_{n\sim}},F_t \rangle : t \in T\}$, where $F_t$ denotes the order filter in $GV_n$ generated by $t \in GV_n$, namely: if $t=i/(n-1)$ then $F_t = \{i/(n-1), (i+1)/(n-1), \ldots, (n-2)/(n-1), 1\}$.\footnote{Strictly speaking, this notation becomes ambiguous if $n$ is not clear from the context and we  identify rational numbers such as $i/(n-1)$ and $i.k/(n-1).k$:  for instance, $1/2$, $2/4$, $3/6$, and so on. In this case, the notation $F_{\frac{1}{2}}$ is problematic, since it  could denote any of an infinite sequence of different filters in  $GV_3$, $GV_5$, $GV_7$,\ldots\ respectively. The right notation for order filters in $GV_n$ should be $F^n_t$. However, the superscript $n$ will be avoided when there is no risk of confusion.} Note that $F_1=\{1\}$. The set of all the logics $L({\cal M}_T )$, for $\emptyset \neq T \subseteq GV_n \setminus\{0\}$, will be denoted by $L({\bf GV_{n \sim}})$. 

 \begin{proposition} \label{propGVn}
The logics $L({\cal M}_{\{t\}})$, with $t \in GV_n \setminus\{0\}$, are pairwise incomparable. Moreover, $L({\cal M}_T )$  is not comparable to  $L({\cal M}_R )$ if $\emptyset\neq T, R \subseteq GV_n \setminus\{0\}$ such that $T\neq R$ and $T$ and $R$ have  the same cardinality.   In addition, the set of logics $L({\bf GV_{n \sim}})$ is a  meet-semilattice  where the logics $L({\cal M}_{\{t\}})$, for $t \in GV_n \setminus\{0\}$, are its maximal elements.
\end{proposition}

\begin{proof}
Let  $\vdash_{\{t\}}$ be the consequence relation of the logic $L({\cal M}_{\{t\}} )$ defined by the matrix $\langle {\bf GV_{n \sim}}, F_t\rangle$, with $ t \in GV_n \setminus\{0\}$. Observe that in any of these logics, since we have the $\Delta$ operator,  it is possible to build a propositional formula on $n$  variables  $\Phi(p_0,p_1,\ldots,p_{n-1})$ such that,  for every evaluation $e$ of formulas on $\bf GV_{n \sim}$,
$$e(\Phi(p_0,p_1,\ldots,p_n)) = \left \{ \begin{tabular}{ll}
$1$, & if $e(p_i)= \frac{i}{n-1}$ for all  $ i= 0,1,\ldots,n-1 $\\
$0$, & otherwise.\\
\end{tabular}\right. $$
Let $i,j\in \{1,2,\ldots,n-1\}$ be such that $i< j$. Then,
\begin{itemize}
\item $\Phi(p_0,p_1,\ldots,p_n) \land p_i \vdash_{\{\frac{j}{n-1}\}} \bot$ and $\Phi(p_0,p_1,\ldots,p_n) \land p_i \nvdash_{\{\frac{i}{n-1}\}} \bot$
\item $\Phi(p_0,p_1,\ldots,p_n) \land p_j \nvdash_{\{\frac{j}{n-1}\}} p_i$ and $\Phi(p_0,p_1,\ldots,p_n) \land p_j\vdash_{\{\frac{i}{n-1}\}} p_i$
\end{itemize}
Therefore $\vdash_{\{t\}}$ and $\vdash_{\{t'\}}$ are not comparable if $0 < t < t' < 1$.
From this, it is easy to prove that for any subsets {$\emptyset\neq T, R \subseteq GV_n \setminus\{0\}$} with the same cardinality and such that $T \neq R$, the logic $L({\cal M}_T )$  is not comparable to  $L({\cal M}_R )$. Finally, if $\emptyset\neq T, R \subseteq GV_n \setminus\{0\}$ then $L({\cal M}_T) \cap L({\cal M}_R) = L({\cal M}_{T \cup R})$. Hence $L({\bf GV_{n \sim}})$ is a  meet-semilattice such that the maximal elements are exactly the logics $L({\cal M}_{\{t\}})$, for $t \in GV_n \setminus\{0\}$.
\end{proof}

On the other hand,  as it was done in \cite{CoEsGo} and \cite{CoEsGiGo} for {\L}ukasiewicz finite-valued logics, product matrices can be considered.

\begin{definition} \label{def_LT} Given a nonempty set $T \subseteq GV_n\setminus\{0\}$, $T=\{t_1,\ldots,t_k\}$ (where $k \geq 1$ and $t_i < t_j$ if $i < j$), we will denote by $\mathbb{L}(T)$ the matrix logic $\langle{\bf (GV_{n\sim})}^k, \Pi_{i =1}^k
F_{t_i}\rangle$  defined on a direct product of {$\bf GV_{n \sim}$} by means of order filters.
\end{definition}

Proceeding as in \cite[Prop. 11]{CoEsGo}, one can show that  $\mathbb{L}(T)$ can be characterized as follows:
\begin{center}
$\Gamma \vdash_{\mathbb{L}(T)} \varphi$ iff \; either $\Gamma \vdash_{\{t_k\}} \bot $ \; or $\Gamma \vdash_{\{t\}} \varphi$ for all $t \in T$.
\end{center}
Note that the first condition, $\Gamma \vdash_{\{t_k\}} \bot$, amounts to a sort of graded inconsistency condition for $\Gamma$ (it reads $e(\psi_i) < \max T$ for any $\bf GV_n$-evaluation $e$ and for any $\psi_i \in \Gamma$). On the other hand, the second condition, $\Gamma \vdash_{\{t\}} \varphi$ for all $t \in T$, amounts to require that $\varphi$ follows from $\Gamma$ in the logic defined by the set matrices ${\cal M}_T = \{\langle {\bf GV_{n\sim}},F_t \rangle : t \in T\}$. Hence, $\varphi$ follows from $\Gamma$ in $\mathbb{L}(T)$ whenever, either $\Gamma$ is inconsistent or contradictory to a certain degree (the maximum of $T$), or $\varphi$ follows from $\Gamma$ in the logic of ${\cal M}_T$. This makes it clear that the latter is a sublogic of the product matrix logic $\mathbb{L}(T)$.

The results become different when studying the matrix logics that involve components over finite subalgebras belonging to the variety generated by $\bf GV_{n \sim}$ because even though all of them are direct products of subalgebras of $\bf GV_{n \sim}$, the number of subalgebras of $\bf GV_{n \sim}$ is significantly larger than in the {\L}ukasiewicz case. Indeed:
 \begin{itemize}
 \item Subalgebras of $\bf GV_{n \sim}$ are those chains that can be obtained from $GV_n$ by removing
 a set of pairs of elements $\{a_i, \ninv a_i\}$ with $a_i \notin \{0,1\}$. In particular, if $n$ is odd one can remove just the fix point.
 \item Therefore the logics between $\Gnsimleq$ and CPL are those logics defined by matrices over direct products of subalgebras of $\bf GV_{n \sim}$ and with products of order filters on the corresponding components of the product algebra. Of course, we have to avoid the repetition of components in these products.
 \end{itemize}

\begin{example}\label{exJ3xJ4}
Recall the logics $\mathsf{J}_3$ and $\mathsf{J}_4$ from Remark \ref{J34}.
Since $\bf GV_{3}$ and $\bf GV_{4}$ are subalgebras of $\bf GV_{5}$, by the characterization of all extensions of $\Gnsimleq$ we have that $\mathsf{J}_3\times \mathsf{J}_4$ coincides (up to language) with $\langle {\bf GV_{3 \sim}} \times {\bf GV_{4 \sim}}, F_{\frac{1}{2}}\times F_{\frac{1}{3}}\rangle$, hence it is a paraconsistent extension of $\Gcsimleq$ that is comparable neither to $\mathsf{J}_3$ nor to $\mathsf{J}_4$. Indeed, it is immediate to see that $\vdash_{\mathsf{J}_3\times \mathsf{J}_4}\varphi$ iff $\vdash_{\mathsf{J}_3}\varphi$ and $\vdash_{\mathsf{J}_4}\varphi$ for every formula $\varphi$. Thus, since the theorems of $\mathsf{J}_3$ and those of $\mathsf{J}_4$ are not comparable,  $\mathsf{J}_3\times \mathsf{J}_4$ is an extension neither of $\mathsf{J}_3$ nor of $\mathsf{J}_4$. 

On the other hand, it is also easy to check that $\varphi \vdash_{\mathsf{J}_3\times \mathsf{J}_4}\bot$ iff $\varphi \vdash_{\mathsf{J}_3}\bot$ or $\varphi \vdash_{\mathsf{J}_4}\bot$. Consider now the formulas:

- $\alpha = \Delta(p \leftrightarrow \ninv p)$

- $\beta = \neg((p_1 \to p_2) \lor (p_2 \to p_3) \lor (p_3 \to p_4))$

\noindent It is clear that $\alpha \not\vdash_{\mathsf{J}_3}\bot$ while  $\alpha \vdash_{\mathsf{J}_4}\bot$, and hence $\alpha \vdash_{\mathsf{J}_3\times \mathsf{J}_4}\bot$ as well. Analogously, we also have that $\beta \not\vdash_{\mathsf{J}_4}\bot$ while  $\beta \vdash_{\mathsf{J}_3}\bot$, and hence $\beta \vdash_{\mathsf{J}_3\times \mathsf{J}_4}\bot$. 
Thus, $\mathsf{J}_3\times \mathsf{J}_4$, $\mathsf{J}_3$ and $\mathsf{J}_4$ are mutually not comparable.
\end{example}

  Finally we can characterize the logics satisfying the explosion rule for $\ninv$:
  $$\frac{\varphi \quad \ninv\varphi}{\bot}$$
 Indeed, we have:  \begin{itemize}
 \item Following the same reasoning as in  \cite{CoEsGo} for the $n$-valued {\L}ukasiewicz case, one can show that  the minimal matrix logic satisfying the explosion rule, i.e.\ the expansion of $\Gnsimleq$ with the above rule, is the logic $L_{exp}$ whose consequence relation is defined as
 \begin{center}
 $\Gamma \vdash_{L_{exp}} \varphi$ iff either $\Gamma \vdash_{\{ \frac{i}{n-1}\}} \bot$,  or $\Gamma \vdash_{\Gnsimleq} \varphi$.
 \end{center}
where $i$ is the first natural such that $\frac{i}{n-1} > 1/2$. By manipulating the right hand-side of the above condition,  $\Gamma \vdash_{L_{exp}} \varphi$ turns out to be equivalent to the  two further conditions:

$\Gamma \vdash_{L_{exp}} \varphi$  iff  either $\Gamma \vdash_{\{ \frac{i}{n-1}\}} \bot$  or ($\Gamma \vdash_{T_1} \varphi$ and $\Gamma \vdash_{T_2} \varphi$) 

\hspace{1.5cm} iff  [$\Gamma \vdash_{\{ \frac{i}{n-1}\}} \bot$  or $\Gamma \vdash_{T_1} \varphi$]  and [$\Gamma \vdash_{\{ \frac{i}{n-1}\}} \bot$  or  $\Gamma \vdash_{T_2} \varphi$)]

\noindent But, according to the paragraph after Def. \ref{def_LT},   the condition [$\Gamma \vdash_{\{ \frac{i}{n-1}\}} \bot$  or $\Gamma \vdash_{T_1} \varphi$] is just saying that $\varphi$ follows from $\Gamma$ in the logic   $\mathbb{L}(T_1)$, while the condition  [$\Gamma \vdash_{\{ \frac{i}{n-1}\}} \bot$  or  $\Gamma \vdash_{T_2} \varphi$)] is clearly equivalent to only 
$\Gamma \vdash_{T_2} \varphi$. Therefore, 
$$ L_{exp} =  \mathbb{L}(T_1) \cap L({\cal M}_{T_2}),$$
or in other words, $L_{exp}$ is the logic defined by the following set of matrices:  

\begin{tabular}{lll}
${\cal M}_n$ & = &$\{ \langle {\bf (GV_{n \sim})}^{i},\Pi_{r = 1}^{i}F_{\frac{r}{n-1}}\rangle \} \; \cup$ \\
& & $ \{ \langle {\bf GV_{n \sim}},F_1\rangle, \langle {\bf GV_{n \sim}}, F_{\frac{n-2}{n-1}}\rangle, \ldots, \langle {\bf GV_{n \sim}}, F_{\frac{i+1}{n-1}}\rangle \}$.
\end{tabular}

 \item Therefore, the explosion rule is valid in all the logics extending the logic $L_{exp}$. Hence, all of them are explosive, while those not extending it are paraconsistent.
 \end{itemize}

  \subsection{Example: the case $n = 5$} \label{n5}
 As an example we study the case of  the set $Int(\Gcsimleq)$ of matrix logics defining intermediate logics between $\Gcsimleq$ and CPL. Recall that $GV_{5}$ denotes the ordered set $\{0,1/4,1/2,3/4,1\}$. We start with some basic facts:

 \begin{itemize}
\item Consider the subset $L({\bf GV_{5\sim}}) \subset Int(\Gcsimleq)$ of logics defined by the set of matrices ${\cal M}_T = \{\langle {\bf GV_{5\sim}}, F_{t}\rangle : t \in T\}$ for $\emptyset \neq T \subseteq GV_5 \setminus \{0\}$, as it was done in Subsection~\ref{Gn}.
According to Proposition~\ref{propGVn}, the logics of the matrices $\langle {\bf GV_{5\sim}}, F_{i/4}\rangle$ for $i \in \{1, 2, 3, 4\}$ are pairwise incomparable, and in fact they are the maximal logics in $L({\bf GV_{5\sim}})$, while $\bigcap_{ i \in \{1, 2, 3, 4\}} \langle {\bf GV_{5\sim}}, F_{i/4}\rangle = \Gcsimleq$ is the minimum logic of $L({\bf GV_{5\sim}})$ (and clearly of $Int(\Gcsimleq)$ as well).

\item Let $L_\Pi({\rm G}_{5 \sim}) \subset Int(\Gcsimleq)$
be the set  of  matrix logics of the form $\mathbb{L}(T)$ defined on direct products of {$\bf GV_{5 \sim}$} by means of products of order filters (recall Definition~\ref{def_LT}). Then, these logics satisfy the following conditions (like in the case of {\L}ukasiewicz logics):

 \begin{itemize}
\item If $\emptyset\neq T, R \subseteq GV_5\setminus\{0\}$ are such that $\max \ T = \max \ R$, then $\mathbb{L}(T) \cap \mathbb{L}(R) = \mathbb{L}(T \cup R)$.

\item The maximal elements of
$L_\Pi({\rm G}_{5 \sim})$ are the matrix logics of the type $\langle {\bf (GV_{5\sim})}^2, F_{i/4} \times F_{j/4} \rangle$ with $i,j \in \{1,2,3,4\}$ and $i < j$.

\item The matrix logic $\langle {\bf (GV_{5 \sim})}^2, F_i \times F_j\rangle$ for $0<i < j$ contains $\langle {\bf GV_{5 \sim}}, F_j\rangle$  and it is not comparable with $\langle {\bf GV_{5 \sim}}, F_k\rangle$ for $0<k \neq j$.

\end{itemize}

\item Finally let us consider the subset $L_{\Pi^*}({\rm G}_{5 \sim})\subset Int(\Gcsimleq)$ of matrix logics defined on direct products of $\bf GV_{5 \sim}$ and their subalgebras together with direct products of order filters.  The subalgebras of $\bf GV_{5 \sim}$ are (isomorphic to)  $\bf GV_{2 \sim}$, $\bf GV_{3 \sim}$ and $\bf GV_{4 \sim}$, and thus the number of matrix logics in $L_{\Pi^*}({\rm G}_{5 \sim})$ proliferate in a large number. Namely, to define matrix logics we have the following components to combine: 4 algebras, $\bf GV_{5 \sim}$, $\bf GV_{4 \sim}$, $\bf GV_{3 \sim}$ and $\bf GV_{2 \sim}$, and 10 order filters:  4 over $\bf GV_{5 \sim}$, 3 over $\bf GV_{4 \sim}$, 2 over $\bf GV_{3 \sim}$ and 1 over $\bf GV_{2 \sim}$. Therefore we have all the possible products (without repetitions) of these 10 components.
 \end{itemize}

  We can also characterize the minimal extension of $\Gcsimleq$ with the explosion rule as the logic $L({\cal M}_5)$ of the set of matrices
 $$ {\cal M}_5 = \{  \langle {\bf (GV_{5 \sim})}^3, F_{3/4} \times  F_{2/4} \times F_{1/4}\rangle, \langle {\bf GV_{5 \sim}}, F_1\rangle \}.$$

Concerning axiomatization, as in case of {\L}ukasiewicz logics, we can give an axiomatic characterization of the logics of $L\Pi({\rm G}_{5 \sim})$.
To see this,
first of all, observe that in G$_{5 \sim}$, for every value $i/4 \in GV_5\setminus \{0\}$ there exists a formula in one variable $\varphi(p)$ characterizing the value $i/4$, i.e.\ such that for any evaluation $e$, $e(\varphi(p)) = 1$ if $e(p) = i/4$, and $0$ otherwise. For example, for the value  $1/2$ the formula can be $\Delta (p \leftrightarrow \ninv p)$. It is also possible to define a formula characterizing the sets of values $\geq i/4, > i/4, \leq i/4$ and $< i/4$.

 Using this observation, it is easy to see that every matrix logic of type  $\langle {\bf GV_{5\sim}}, F_{i/4}\rangle$ or $\mathbb{L}(T) \in L\Pi({\rm G}_{5 \sim})$ can be axiomatized. For instance, here we give the following example:
 \begin{itemize}
 \item The matrix logic $({\bf GV_{5\sim}}, F_{i/4})$ is  axiomatized by adding to the axioms and rules of  $\Gcsimleq$ the following restricted inference rule:
 \begin{description}
\item[] if $\vdash_{{\rm G}_{5\sim}} (\varphi < i/4)\lor ((\varphi \geq i/4) \land (\psi \geq i/4)$,  from $\varphi$ derive $\psi$
\end{description}
 \end{itemize}
Other matrix logics of $L\Pi({\rm G}_{5 \sim})$ can be axiomatized in an analogous way.
 Notice that these axiomatizations are possible since, in G$_{5 \sim}$, for every element $a \in GV_5$ there exists a characterizing formula in one variable. This is not true in G$_{n \sim}$ for $n>5$, and thus the previous axiomatization results are not generalizable to G$_{n \sim}$ for $n > 5$.

\section{Ideal and saturated paraconsistent extensions of $\mathrm{G}\mystrutsec^{\text{\footnotesize $\leq$}}_{n\sim}$} \label{ideal}

 As already noticed,  matrix logics over direct products of subalgebras of $\bf GV_{n \sim}$ with products of order filters are $\sim$-paraconsistent iff all the components are $\sim$-paraconsistent. In this section, using the results of the previous section, we study the status of the logics between $\Gnsimleq$ and CPL in relation to ideal $\sim$-paraconsistency. Namely, we  show that there are only two extensions of $\Gnsimleq$ which are ideal $\sim$-paraconsistent. Moreover we show that there is another $\sim$-paraconsistent extension of $\Gnsimleq$  which, although not being ideal $\sim$-paraconsistent, it has the remarkable property of not having any proper $\sim$-paraconsistent extension.

We have already briefly discussed in the Introduction  the concept of {\em ideal paraconsistent logics},  introduced by Arieli et al.  in   \cite{ArieliAZ11a}.\footnote{The authors, as it was mentioned in Section~\ref{sect1},  have changed the terminology ``ideal paraconsistent logic'' to ``fully maximal and normal paraconsistent logic''. However, it should be noticed that being normal, according to~\cite[Definition~1.32]{AAZbook}, means that the logic $L$ has, besides a deductive implication, a conjunction and a disjunction satisfying the usual properties in terms of consequence relations. Here we decide to keep the original definition of ideal paraconsistency. It is worth noting that all the ideal (and saturated) logics considered in this paper and in~\cite{CoEsGiGo} are normal in the sense of~\cite{AAZbook}.} We recall here this notion.

\begin{definition} [c.f. \cite{ArieliAZ11a}] \label{IdPar}
Let $L$ be a propositional logic defined over a signature $\Theta$ (with consequence relation $\vdash_L$) containing at least a unary connective $\urcorner$ and a binary connective $\to$ such that:
\begin{itemize}
\item[(i)] $L$ is paraconsistent w.r.t.\ $\urcorner$ (or simply $\urcorner$-paraconsistent), that is,  there are formulas $\varphi,\psi \in \mathcal{L}(\Theta)$ such that $ \varphi, \urcorner\varphi, \nvdash_L\psi$;
\item[(ii)] $\to$ is an implication for which the deduction-detachment theorem holds in $L$, that is, $\Gamma \cup \{\varphi\} \vdash_L \psi$ iff  $\Gamma \vdash_L \varphi \to\psi$, for every set for formulas $\Gamma \cup \{\varphi, \psi\} \subseteq  \mathcal{L}(\Theta)$.
\item[(iii)] There is a presentation of CPL as a matrix logic  $L' = \langle \mathbf{A}, \{1\}\rangle$  over the signature $\Theta$ such that the domain of $\mathbf{A}$ is $ \{0,1\}$, and $\urcorner$ and $\to$ are interpreted as the usual 2-valued negation and implication of CPL, respectively.

\item[(iv)] $L$ is a sublogic of CPL in the sense that ${\vdash_L} \subseteq {\vdash_{L'}}$, that is, $\Gamma \vdash_L \varphi$ implies  $\Gamma \vdash_{L'} \varphi$,  for every set for formulas $\Gamma \cup \{\varphi\} \subseteq  \mathcal{L}(\Theta)$.
\end{itemize}
Then, $L$ is said to be an {\em ideal $\urcorner$-paraconsistent logic} if it is maximal w.r.t. CPL, and every proper extension of $L$ over $\Theta$ is not $\urcorner$-paraconsistent.

\end{definition}

An implication connective satisfying the above condition (ii) is usually called {\em deductive implication}.

\begin{remark}\label{3i4IdealParac}
As it has been argued in Remark \ref{J34},
$\mathsf{J}_3$ is equivalent to $\langle {\bf GV}_{3\sim}, F_{\frac{1}{2}}\rangle$ and therefore for every odd number $n\geq3$, $\mathsf{J}_3$ is an extension of any $\Gnsimleq$ (recall Fig. \ref{G_sim}).
Similarly,  $\mathsf{J}_4$ is equivalent to $\langle {\bf GV}_{4\sim}, F_{\frac{1}{3}}\rangle$. Thus $\mathsf{J}_4$ is an extension of $\Gnsimleq$  for every $n\geq 4$.

In \cite[Proposition 6.3]{CoEsGiGo} it is shown that $\mathsf{J}_3$ and $\mathsf{J}_4$ are ideal $\sim$-paraconsistent logics where the deductive implication in the signature of $\rm{G}_{\sim}$
is the term-defined implication $x\Rightarrow y:=  \neg x \lor y$.\footnote{Observe that in~\cite{CoEsGiGo}, $\neg$ denotes the \L ukasiewicz negation, while the G\"odel negation for $\mathsf{J}_3$ and $\mathsf{J}_4$  is respectively denoted by $\sim^1_2$ and $\sim^1_3$.}
\end{remark}

As discussed in Section~\ref{sect1}, requiring a paraconsistent logic to   be maximal w.r.t. CPL in order to be `ideal' (in the sense of being `optimal') is a debatable  issue (see~\cite{WansOdin}). On the other hand, the other requirements of Definition~\ref{IdPar} seem interesting, and a system enjoying such features could be considered as a remarkable paraconsistent logic. This motivates the following definition.

\begin{definition}
 A logic $L$ is {\em saturated} $\urcorner$-paraconsistent  if it satisfies all the conditions $(i)$ to $(iv)$ of the previous definition, and every proper extension of $L$ over $\Theta$ is not $\urcorner$-paraconsistent.\footnote{Using the terminology of~\cite{AAZbook}, a saturated paraconsistent logic is a logic such that: (i)~it has an implication, (ii)~it is {\bf F}-contained in CPL,  and (iii)~it is strongly maximal.}
\end{definition}

{
\begin{remark}
In \cite[p. 273]{RibCon} it was introduced the notion of maximality of a logic $L$ w.r.t. an inference rule $r$. Namely, given a  Tarskian and structural propositional logic $L$ over a signature $\Theta$, and given an inference rule $r$ over $\Theta$, $L$ is {\em $r$-maximal} if $r$ is not derivable in $L$, but any proper extension of $L$ over $\Theta$ derives $r$.\footnote{This was denoted by $L \in {\bf Triv}_\Theta \bot \{r\}$ in \cite{RibCon}, where ${\bf Triv}_\Theta$ denotes the trivial logic over the signature $\Theta$.} Clearly ideal and saturated paraconsistent logics are special cases of $r$-maximal logics, where $r$ is the explosion rule.\footnote{Indeed, by means of the notion of {\em remainder set} $L \bot R$ of a logic $L$ w.r.t. a set of rules $R$ introduced in~\cite[Definition~7]{RibCon}, several concepts relative to  maximality proposed in the literature can be easily represented, see~\cite[p. 273]{RibCon}.}

\end{remark}}

\begin{proposition}\label{J3xJ4}
  $\mathsf{J}_3\times \mathsf{J}_4:=\langle {\bf GV_{3\sim}\times GV_{4\sim}}, F_{\frac{1}{2}}\times F_{\frac{1}{3}}\rangle$ is saturated $\sim$-paracon\-sis\-tent, but not ideal $\sim$-paraconsistent.
\end{proposition}
\begin{proof}

Since $\bf GV_{3\sim}$ and $V\bf \rm{G}_{4\sim}$ are subalgebras of $\bf GV_{5\sim}$, by the characterization of all extensions of $\Gnsimleq$ given in subsection \ref{Gn}, $\langle {\bf GV_{3\sim}\times GV_{4\sim}}, F_{\frac{1}{2}}\times F_{\frac{1}{3}}\rangle$ is an extension of $\Gcsimleq$ satisfying conditions $(i)$ to $(iv) $ because every component is $\sim$-paraconsistent and $x\Rightarrow y:= \lnot x \lor y$ is a term-defined  deductive implication. We prove by contradiction that $\mathsf{J}_3\times \mathsf{J}_4$ has no proper $\sim$-paraconsistent extensions. Assume there is  a proper $\sim$-paraconsistent extension $L$ of $\mathsf{J}_3\times \mathsf{J}_4$.  In this case there is a matrix $\langle {\bf A_{1}\times\cdots\times A_{k}}, F_{i_{1}}\times\cdots\times F_{i_{k}}\rangle$ which is an extension of $L$ such that each $\langle {\bf A_{j}},F_{i_{j}}\rangle$ is either $\mathsf{J}_3$, $\mathsf{J}_4$, $\langle {\bf GV_{5\sim}}, F_{\frac{1}{2}}\rangle$ or $\langle {\bf GV_{5\sim}}, F_{\frac{1}{4}}\rangle$. Since  $\mathsf{J}_3$ is not comparable with $\mathsf{J}_3\times \mathsf{J}_4$ and $\mathsf{J}_3$ is a submatrix of $\langle {\bf GV_{5}}, F_{\frac{1}{2}}\rangle$ and also a submatrix of $\langle {\bf GV_{5\sim}}, F_{\frac{1}{4}}\rangle$, there is a component $\langle {\bf A_{j0}},F_{j0}\rangle=\mathsf{J}_4$. Similarly, there should be a  different component $\langle {\bf A_{j1}},F_{j1}\rangle\neq \mathsf{J}_4$, otherwise $\mathsf{J}_4$ would be an extension of $\mathsf{J}_3\times \mathsf{J}_4$. Finally, in the case $\langle {\bf A_{1}\times\cdots\times A_{k}}, F_{i_{1}}\times\cdots\times F_{i_{k}}\rangle$ has a component equal to $\mathsf{J}_4$ and another which is different to $\mathsf{J}_4$,  then  $\mathsf{J}_3\times \mathsf{J}_4$ is a submatrix of $\langle {\bf A_{1}\times\cdots\times A_{k}}, F_{i_{1}}\times\cdots\times F_{i_{k}}\rangle$, which contradicts the fact that $L$ is a proper extension of $\mathsf{J}_3\times \mathsf{J}_4$.

Let $\varphi$ be a theorem of $\mathsf{J}_3$ which is not a theorem of $\mathsf{J}_4$. Then, the matrix logic $J_{2}\times \mathsf{J}_3:=\langle {\bf GV_{2\sim}\times GV_{3\sim}}, F_{1}\times F_{\frac{1}{2}}\rangle$ is an extension of $\mathsf{J}_3\times \mathsf{J}_4$ different from CPL such that $\vdash_{J_{2}\times \mathsf{J}_3}\varphi$. Thus $\mathsf{J}_3\times \mathsf{J}_4$ is not maximal w.r.t. CPL.
\end{proof}

\begin{theorem}\label{CarAlIdParGn}
Let $n$ be an integer number such that $n>4$ and let $L$ be an extension of $\Gnsimleq$.
\begin{enumerate}
\item If $n$ is an even number,  the following are equivalent:
\begin{itemize}
\item[-] $L$ is saturated $\sim$-paraconsistent 
\item[-] $L$ is ideal $\sim$-paraconsistent 
\item[-] $L=\mathsf{J}_4$
\end{itemize}
\item If $n$ is an odd number,   the following are equivalent: 
\begin{itemize}
\item[-] $L$ is saturated  $\sim$-paraconsistent
\item[-] $L=\mathsf{J}_3$, $L=\mathsf{J}_4$ or $L=\mathsf{J}_3\times \mathsf{J}_4$.
\end{itemize}
\item If $n$ is an odd number,    the following are equivalent: 
\begin{itemize}
\item[-] $L$ is ideal  $\sim$-paraconsistent 
\item[-] $L=\mathsf{J}_3$ or $L=\mathsf{J}_4$.
\end{itemize}
\end{enumerate}
\end{theorem}
\begin{proof}

\begin{enumerate}
\item {Assume that $n$ is even. After Remark \ref{3i4IdealParac} and Proposition \ref{J3xJ4} we only need to prove that if $L$ is saturated $\sim$-paraconsistent then $L=\mathsf{J}_4$.  Since $n$ is even then, as observed in Subsection~\ref{Gn}, every extension $L'$ of $\Gnsimleq$ is induced by a family of matrices of the form $\langle {\bf A},F\rangle=\langle {\bf GV_{n_{1}\sim}\times \cdots \times GV_{n_{k}\sim}}, F_{\frac{i_{1}}{n_{1}-1}}\times\cdots\times F_{\frac{i_{k}}{n_{k}-1}}\rangle$ where each $n_j$ is also an even number.} If $L'$ is $\sim$-paraconsistent then  there is a member in that family, say $\langle {\bf GV_{n_{1}\sim}\times \cdots \times GV_{n_{k}\sim}}, F_{\frac{i_{1}}{n_{1}-1}}\times\cdots\times F_{\frac{i_{k}}{n_{k}-1}}\rangle$, such that  $2<n_{j}\leq n$ and $0<\frac{i_{j}}{n_{j}-1}\leq \frac{1}{2}$ for every $j$ such that $1\leq j\leq k$. Then,  $\mathsf{J}_4$ is an extension of every $\sim$-paraconsistent extension of $\Gnsimleq$. In particular, $J_4$ extends $L$. Thus $L=\mathsf{J}_4$, since $L$ is maximal paraconsistent.

\item The right to left implication follows from Remark \ref{3i4IdealParac} and Proposition \ref{J3xJ4}. To prove the converse, let $L$ be an saturated $\sim$-paraconsistent extension of $\Gnsimleq$. Since  $L$ is $\sim$--paraconsistent and it has no proper $\sim$--paraconsistent extension, $L$ is induced by a single $\sim$--paraconsistent matrix $\langle {\bf A},F\rangle$ such that $\langle {\bf A},F\rangle=\langle {\bf GV_{n_{1}\sim}\times \cdots \times GV_{n_{k}\sim}}, F_{\frac{i_{1}}{n_{1}-1}}\times\cdots\times F_{\frac{i_{k}}{n_{k}-1}}\rangle$ where $2<n_{j}\leq n$ and $0<\frac{i_{j}}{n_{j}-1}\leq \frac{1}{2}$ for every $j$ such that $1\leq j\leq k$.
    If all $n_{j}$'s are even, as in previous item  $L=\mathsf{J}_4$.
    If all $n_{j}$'s are odd, then $\mathsf{J}_3$ is a $\sim$--paraconsistent extension of $L$, thus $L=\mathsf{J}_3$.
    Assume $n$ is odd and some $n_{j}$'s are even and some are odd, all of them bigger than $2$. Then in that case $\mathsf{J}_3\times \mathsf{J}_4:=\langle {\bf GV_{3\sim}\times GV_{4\sim}}, F_{\frac{1}{2}}\times F_{\frac{1}{3}}\rangle$ is a $\sim$--paraconsistent extension of $L$, thus $L=\mathsf{J}_3\times \mathsf{J}_4$.
\item Immediate after Proposition \ref{J3xJ4} and item 2.
\end{enumerate}
\end{proof}

\section{Saturated paraconsistency and finite-valued\\ {\L}ukasiewicz logics} \label{sect-luka}

In \cite{CoEsGiGo} we
study maximality conditions for intermediate logics between CPL and the degree-preserving finite-valued {\L}ukasiewicz logics. In particular we have characterized the ideal paraconsistent logics in this family. Since in the last section we have introduced the weaker notion of saturated paraconsistency in the setting of degree-preserving G\"odel logics with involution, we deem interesting to also explore this notion for the above mentioned setting of finite-valued {\L}ukasiewicz logics. This is done in this section, after briefly recalling some basic notions about (degree-preserving) finite-valued {\L}ukasiewicz logics.

The $(n+1$)-valued {\L}ukasiewicz logic can be semantically defined as the matrix logic
$$\L_{n+1}=\langle \algLfin_{n+1}, \{1\}  \rangle, $$
where {$\algLfin_{n+1} = (\domLfin_{n+1}, \neg, \to)$} is the $n+1$-elements MV-chain   with $\domLfin_{n+1} = \big\{0,\frac{1}{n},\dots,\frac{n-1}{n},1\big\}$, and operations  defined as follows: for every $x,y \in \domLfin_{n+1}$,
\begin{itemize}
\item[] $\neg x =1-x$ ({\L}ukasiewicz negation)
\item[] $x \to y = \min \{1, 1-x+y\}$ ({\L}ukasiewicz implication)
\end{itemize}
In fact $\L_{n+1}$ is algebraizable and its generated quasi-variety $MV_{n+1}:=\mathcal{Q}(\algLfin_{n+1})$ is its equivalent algebraic semantics.

The (finitary) degree preserving $(n+1$)-valued {\L}ukasiewicz logic, denoted $\L_{n+1}^{\leq}$, can be semantically defined the following way: For every finite set of formulas $\Gamma\cup\{\varphi\}$\\

\begin{tabular}{lll}
$\Gamma\models_{\textrm{\L}_{n+1}^{\leq}} \varphi$ & if & for every evaluation $e$ over $\algLfin_{n+1}$ and every $a\in \domLfin_{n+1}$,\\
&&  if $e(\gamma)\geq a$ for every $\gamma\in \Gamma$, then $e(\gamma)\geq a$.\\
\end{tabular}
\mbox{} \\

Following \cite{CoEsGiGo} we denote by $\mathsf{L}_{n}^{i}$  the logic obtained by the matrix $\langle \algLfin_{n+1}, F_{\frac{i}{n}}  \rangle$, where $F_{\frac{i}{n}}$ is the order filter $\{ x \in \domLfin_{n+1} : x \geq i/n \}$.  Notice that with this notation the $n+1$-valued {\L}ukasiewicz logic $ \L_{n+1}$ is also denoted by $\mathsf{L}_{n}^{n}$.

As proved in \cite[Theorem 5.2]{CoEsGiGo}, for every $1\leq i \leq n$, $\mathsf{L}_{n}^{i}$  is equivalent, as a deductive system, to $ \L_{n+1}$ (see \cite{blok:pig:01} and also \cite{Blok-Pigozzi:DeductionTheorems}). Since algebraizability is preserved by equivalence, $\mathsf{L}_{n}^{i}$ is algebraizable and $MV_{n+1}$ is also its equivalent algebraic semantics. Thus, the lattice of finitary extensions of  $\mathsf{L}_{n}^{i}$ is isomorphic to the lattice of subquasivarieties of $MV_{n+1}=\mathcal{Q}(\algLfin_{n+1})$.

$MV_{n+1}$ is a locally finite variety and, as proved in \cite{GT14}, every subquasivariety is also locally finite and it is generated by a finite family of critical \footnote{An algebra is said to be {\em critical} if it is a finite algebra not belonging to the quasivariety generated by all its proper subalgebras.} MV-algebras.   Using the correspondence among subquasivarieties of $MV_{n+1}$ and finitary extensions  of   $\mathsf{L}_{n}^{i}$, in \cite{CoEsGiGo} we obtain that every extension $L$ of $\mathsf{L}_{n}^{i}$ is induced by a finite family  of matrices of type $\langle {\bf A},F\rangle$ where $\bf A$ is a critical $MV_{n+1}$-algebra  and $F$ is a lattice filter of $\bf A$ compatible with $L$. To be more precise, in \cite{GT14,CoEsGiGo} it is proved that $\bf A$ is isomorphic to a direct product of $MV_{n+1}$-chains $\algLfin_{n_{0}+1}\times\cdots\times\algLfin_{n_{l-1}+1}$,  where
\begin{enumerate}
 \item for every $j<l$, $n_{j}|n$
 \item For every $j, k<l$, $ k\neq j$ implies $ n_{k}\neq n_{j}$.
 \item If there exists $n_{j}$, $j<l$ such that $n_{k}|n_{j}$ for some
  $k\neq j$, then $n_{j}$ is unique.
\end{enumerate}
and $F=(F_{\frac{i}{n}})^{l}\cap (\domLfin_{n_{0}+1}\times\cdots\times \domLfin_{n_{l-1}+1})$.

Thus, in analogy to \cite[Theorem 3]{CoEsGo}, every extension of $\mathsf{L}_{n}^{i}$ is induced by a finite family of matrices where each matrix is a product of submatrices of $\langle \algLfin_{n+1}, F_{\frac{i}{n}}  \rangle$.

As observed in Proposition \ref{3i4},  
$\algLfin_{3}$ is {termwise} equivalent
to ${\bf GV}_{3 \sim}$ and $\algLfin_{4}$ is {termwise} equivalent to ${\bf GV}_{4 \sim}$, where the involutive negation $\sim$ in ${\bf GV}_{3\sim}$ and ${\bf GV}_{4\sim}$ is in fact the MV-negation $\neg$. Then, as indicated in   Remark~\ref{3i4IdealParac}, the matrix logics $\mathsf{J}_{3}= \langle \algLfin_{3}, F_{\frac{1}{2}} \rangle$ and $\mathsf{J}_{4}=\langle\algLfin_{4}, F_{\frac{1}{3}}\rangle$, expressed in the signature of {\L}ukasiewicz logic, are ideal $\neg$-paraconsistent. We recall here that this can be generalized in the following way.

\begin{proposition}\label{SufIdParLuk} \emph{\cite[Proposition 6.2 ]{CoEsGiGo}}
Let  $q$ be a prime number, and let $1 \leq i < q$ such that $i/q \leq 1/2$. Then,  $\mathsf{L}^i_q$ is a $(q+1)$-valued ideal $\neg$-paraconsistent logic.
\end{proposition}

In fact, we can also prove that the converse implication also holds under some circumstances. This result  is not present in \cite{CoEsGiGo}.

\begin{theorem}\label{teo-ideal}
Let $0<i<n$ such that $\frac{i}{n}\leq\frac{1}{2}$. If $L$ is an extension of  $\mathsf{L}_{n}^{i}$, then,
$L$ is ideal $\neg$-paraconsistent iff $L=\mathsf{L}_{q}^{j}$ for some prime number $q$ such that $q|n$ and some $1\leq j$ such that $j/q \leq 1/2$
\end{theorem}

\begin{proof}
Let $L$ be an ideal  $\neg$-paraconsistent extension of $\mathsf{L}_{n}^{i}$. Since $L$ is maximal, it is induced by a single matrix $\langle{\bf A}, F\rangle$, where $\bf A$ is critical and $F$ is compatible with $L$. In fact, as mentioned above, $\langle{\bf A}, F\rangle$ is of type $\langle\algLfin_{n_{1}+1}\times\cdots\times\algLfin_{n_{k}+1}, (F_{\frac{i}{n}})^{k}\cap (\domLfin_{n_{1}+1}\times\cdots\times \domLfin_{n_{k}+1})\rangle$ where
\begin{enumerate}
 \item for every $1\leq i\leq k$, $n_{i}|n$
 \item For every $1\leq i, j\leq k$, $ i\neq j$ implies $ n_{i}\neq n_{j}$.
 \item If there exists $n_{j}$, $1\leq j\leq k$ such that $n_{i}|n_{j}$ for some
  $1\leq i\neq j$, then $n_{j}$ is unique.
\end{enumerate}
Since $L$ is $\neg$-paraconsistent, none of the components $\algLfin_{n_{i}+1}$ can be $\algLfin_{2}$ (otherwise $L$ would be explosive), and hence $1<n_{i}$ for all $1\leq i\leq k$.
If $k>1$, then

\begin{itemize}
 \item If there is $n_{j}$, with $1\leq j\leq k$, such that $n_{i}|n_{j}$ for some
  $1\leq i\neq j$, then without loss of generality assume that $j=k$. In that case $\langle \algLfin_{n_{1}+1}\times\cdots\times\algLfin_{n_{k-1}+1}, (F_{\frac{i}{n}})^{k-1}\cap (\domLfin_{n_{1}+1}\times\cdots\times \domLfin_{n_{k-1}+1})\rangle$ is a $\neg$-paraconsistent extension of $L$ which contradicts the assumption of $L$ being ideal $\neg$-paraconsistent.
 \item If there is no $n_{j}$, with $1\leq j\leq k$, such that $n_{i}|n_{j}$ for some
  $1\leq i\neq j$, then $n_{k}\neq n$ and $L$ is not maximal because $\langle \algLfin_{2}\times\algLfin_{n_{k}+1}, (F_{\frac{i}{n}})^{2}\cap (\domLfin_{2}\times \domLfin_{n_{k}+1})\rangle$ is an extension of $L$ different from CPL and there is a formula $\varphi$ valid in $\langle \algLfin_{2}\times\algLfin_{n_{k}+1}, (F_{\frac{i}{n}})^{2}\cap (\domLfin_{2}\times \domLfin_{n_{k}+1})\rangle$ and not valid in $L$. A contradiction again.
  \end{itemize}
  Thus $k=1$. In that case $n$ should be prime because otherwise for any prime number $p$ such that $p|n$, $\langle\algLfin_{p+1}, F_{\frac{i}{n}}\cap \domLfin_{p+1}\rangle$ would be an extension of $L$ different from CPL and there is a formula $\varphi$ valid in $\langle\algLfin_{p+1}, F_{\frac{i}{n}}\cap \domLfin_{p+1}\rangle$ and not valid in $L$.
 \end{proof}

As regards saturated paraconsistency we have the following results:

\begin{theorem} \label{teorX}
Let $0<i<n$ such that $\frac{i}{n}\leq\frac{1}{2}$. Let
$$X=\big\{p \ : \  p \mbox{ prime, } p|n,\   F_{\frac{i}{n}}\cap \domLfin_{p+1}= \big\{\frac{m}{p} \ : \ m\geq k\big\} \mbox{ and } \frac{k}{p}\leq\frac{1}{2}\big\}.$$
For every finite subset $\{p_{1},\ldots, p_{j}\}\subseteq X$, the logic defined by the matrix
$$\langle \algLfin_{p_{1}+1}\times\cdots\times\algLfin_{p_{j}+1}, (F_{\frac{i}{n}})^{j}\cap (\domLfin_{p_{1}+1}\times\cdots\times \domLfin_{p_{j}+1}) \rangle$$
is saturated $\neg$-paraconsistent.
\end{theorem}

\begin{proof}
By the previous result,$\langle \algLfin_{p_{1}+1}\times\cdots\times\algLfin_{p_{j}+1}, (F_{\frac{i}{n}})^{j}\cap (\domLfin_{p_{1}+1}\times\cdots\times \domLfin_{p_{j}+1}) \rangle$ is an extension of $\mathsf{L}_{n}^{i}$. Moreover, it is $\neg$-paraconsistent, because every component is $\neg$-paraconsistent. Let $\Rightarrow^{i}_{n}$ defined as $\varphi\Rightarrow^{i}_{n}\psi:=\ninv^{i}_{n}\varphi \lor \psi$ where $\ninv^{i}_{n}(x)$ is the single variable McNaughton term such that for every $a\in \domLfin_{n+1}$,
$$
\ninv^{i}_{n}(a)  = \left \{
\begin{array}{ll}
0, & \mbox{if } a\geq \frac{i}{n} \\
1, & \mbox{otherwise}.
\end{array}
\right .
$$  Similarly to the proof of \cite[Proposition 6.2]{CoEsGiGo}, the logic $\langle \algLfin_{p_{1}+1}\times\cdots\times\algLfin_{p_{j}+1},$ \\ $ (F_{\frac{i}{n}})^{j}\cap (\domLfin_{p_{1}+1}\times\cdots\times \domLfin_{p_{j}+1})\rangle$ satisfies  conditions (i) to (iv) in  Definition~\ref{IdPar}, the definition of ideal $\urcorner$-paraconsistency. Let $L$ be a $\neg$-paraconsistent extension of $\langle \algLfin_{p_{1}+1}\times\cdots\times\algLfin_{p_{j}+1}, (F_{\frac{i}{n}})^{j}\cap (\domLfin_{p_{1}+1}\times\cdots\times \domLfin_{p_{j}+1}) \rangle$, then $L$ is induced by a finite family of matrices $\langle {\bf A}, F\rangle$, where $\bf A$ is critical and $F$ is compatible with $L$. Since $L$ is $\neg$-paraconsistent, there is at least one matrix $\langle \algLfin_{n_{0}+1}\times\cdots\times\algLfin_{n_{l-1}+1}, (F_{\frac{i}{n}})^{l}\cap (\domLfin_{n_{0}+1}\times\cdots\times \domLfin_{n_{l-1}+1})\rangle$ where
\begin{enumerate}
 \item for every $m<l$, $n_{m}|n$
 \item for every $m, k<l$, $ k\neq m$ implies $ n_{k}\neq n_{m}$
 \item if there exists $n_{m}$ with $m<l$ such that $n_{m}|n_{k}$ for some
  $k\neq m$, then $n_{k}$ is unique,
\end{enumerate}
 which is $\neg$-paraconsistent. Thus for every $m<l$, it is the case that $2\leq n_{m}$. Since $\langle\algLfin_{n_{0}+1}\times\cdots\times\algLfin_{n_{l-1}+1}, (F_{\frac{i}{n}})^{l}\cap (\domLfin_{n_{0}+1}\times\cdots\times \domLfin_{n_{l-1}+1})\rangle$ is an extension of $\langle\algLfin_{p_{1}+1}\times\cdots\times\algLfin_{p_{j}+1}, (F_{\frac{i}{n}})^{j}\cap (\domLfin_{p_{1}+1}\times\cdots\times \domLfin_{p_{j}+1})\rangle$, then  $\langle \algLfin_{n_{0}+1}\times\cdots\times\algLfin_{n_{l-1}+1}, (F_{\frac{i}{n}})^{l}\cap (\domLfin_{n_{0}+1}\times\cdots\times \domLfin_{n_{l-1}+1})\rangle$ is a submatrix of $\langle \algLfin_{p_{1}+1}\times\cdots\times\algLfin_{p_{j}+1}, (F_{\frac{i}{n}})^{j}\cap (\domLfin_{p_{1}+1}\times\cdots\times \domLfin_{p_{j}+1})\rangle$. Therefore, by \cite[Lemma 5.6]{CoEsGiGo},
 for every $m<l$ there is a $0<k\leq j$ such that $n_{m}|p_{k}$, since $2\leq n_{m}$ and $p_{k}$ is prime, then $n_{m}=p_{k}$. Moreover for every $0<k\leq j$, there is $m<l$ such that $n_{m}|p_{k}$. Thus $\algLfin_{p_{1}+1}\times\cdots\times\algLfin_{p_{j}+1}\cong \algLfin_{n_{0}+1}\times\cdots\times\algLfin_{n_{l-1}+1}$ and $L= \langle \algLfin_{p_{1}+1}\times\cdots\times\algLfin_{p_{j}+1}, (F_{\frac{i}{n}})^{j}\cap (\domLfin_{p_{1}+1}\times\cdots\times \domLfin_{p_{j}+1})\rangle$, proving that any proper extension of $\langle \algLfin_{p_{1}+1}\times\cdots\times\algLfin_{p_{j}+1}, (F_{\frac{i}{n}})^{j}\cap (\domLfin_{p_{1}+1}\times\cdots\times \domLfin_{p_{j}+1})\rangle$ is not $\neg$-paraconsistent.
 \end{proof}

\begin{remark} \label{not-ideal}
One may wonder whether the saturated $\neg$-paraconsistent logics identified in the above theorem are in fact ideal paraconsistent. However, it is easy to see that this is not the case unless they are of the type described in Theorem \ref{teo-ideal}. Indeed, this is a consequence of the fact that the logics considered in Theorem \ref{teorX} (and in Corollary \ref{coro} below) are extensions of logics of the type $\mathsf{L}_{n}^{i}$, and in  Theorem \ref{teo-ideal} we have exactly characterized  which of these extensions are ideal paraconsistent.
\end{remark}

\begin{corollary}\label{coro}
Let $\{p_{1},\ldots, p_{j}\}$ be any finite set of prime numbers, then $\langle \algLfin_{p_{1}+1}\times\cdots\times\algLfin_{p_{j}+1}, F_{\frac{1}{p_{1}}}\times\cdots\times F_{\frac{1}{p_{j}}}\rangle$ is saturated $\neg$-paraconsistent.
\end{corollary}

Contrary to the case of $\rm{G}_{n \sim}^\leq$ in Theorem \ref{CarAlIdParGn}, not every saturated $\neg$-paracon\-sis\-tent extension of $\mathsf{L}_{n}^{i}$ is of the type of the above corollary. For instance $\mathsf{L}_{15}^{7}$ is saturated $\neg$-paraconsistent. Indeed, it is a $\neg$-paraconsistent logic where $\Rightarrow^{7}_{15}$ is a deductive implication and $\mathsf{L}_{1}^{1}=$ CPL is a submatrix logic of $\mathsf{L}_{15}^{7}$. Moreover, every proper extension $L$ of $\mathsf{L}_{15}^{7}$ is induced by a family of proper submatrices of $\mathsf{L}_{15}^{7}$, of type $\langle \algLfin_{n_{0}+1}\times\cdots\times\algLfin_{n_{l-1}+1}, (F_{\frac{7}{15}})^{l}\cap (\domLfin_{n_{0}+1}\times\cdots\times \domLfin_{n_{l-1}+1}) \rangle$ where at least there is some $j<l$ such that $n_{j}|15$ and $n_{j}\neq 15$. Hence, $n_{j}$ is either $1$, $3$ or $5$, in which case the $j$-th component $\langle \algLfin_{n_{j}+1},F_{\frac{7}{15}}\cap \domLfin_{n_{j}+1} \rangle$ is not $\neg$-paraconsistent. Thus $L$ is not $\neg$-paraconsistent and,  therefore, $\mathsf{L}_{15}^{7}$ is saturated $\neg$-paraconsistent.

\section{A final remark: relationship to logics of formal inconsistency}

{To conclude this section, we provide an additional analysis --from the point of view of paraconsistency-- of the logics discussed in this paper.
Recall from  Section~\ref{sect1} the class of paraconsistent logics known as {\em logics of formal inconsistency} ({\bf LFI}s). It is easy to see that all the paraconsistent logics considered in the present paper are in fact {\bf LFI}s.

Indeed, in~\cite{EEFGN} it is shown that, if $L_{\ninv}$ is the expansion of a core fuzzy logic $L$ with an involutive negation $\ninv$ where  $\Delta$ is definable in $L_{\ninv}$,\footnote{This is the case of any pseudo-complemented logic $L$ where $\Delta$ is definable as $\Delta \varphi := \neg \ninv \varphi$, in particular the case of G\"odel fuzzy logic $\rm{G}$.}
then $L^{\leq}_{\ninv}$ is an {\bf LFI} w.r.t. $\ninv$, and the consistency operator is given by ${\circ}\varphi = \Delta(\neg \varphi \vee \varphi)$. This shows the following.

\begin{proposition}
All the paraconsistent logics based on G\"odel fuzzy logic with involution $\rm{G}_{\ninv}$ and its finite-valued extensions $\rm{G}_{n\ninv}$ considered in this paper are {\bf LFI}s w.r.t. $\ninv$.
\end{proposition}

As for the paraconsistent logics based on finite-valued {\L}ukasiewicz logics analyzed in this section, they are also {\bf LFI}s, as the following result states. }

\begin{proposition}
Let $L$ be one of the matrix logics in Proposition~\ref{SufIdParLuk}, or one of the products of matrix logics in Theorem~\ref{teorX}. Then, $L$ is an {\bf LFI} w.r.t. $\neg$.

\end{proposition}
\begin{proof}
Concerning the logics of Proposition~\ref{SufIdParLuk},
by~\cite[Proposition~6.3]{CoEsGiGo} we know that each logic $\mathsf{L}^i_n$ for $i/n \leq 1/2$ is an {\bf LFI} w.r.t. $\neg$, where the consistency operator is given by $\circ^i_n \alpha := {\sim}^i_n (\alpha \wedge \neg \alpha)$. Here, ${\sim}^i_n$ is the  unary connective defined as in the proof of Theorem~\ref{teorX}. Now, let
$$L=\langle \algLfin_{p_{1}+1}\times\cdots\times\algLfin_{p_{j}+1}, (F_{\frac{i}{n}})^{j}\cap (\domLfin_{p_{1}+1}\times\cdots\times \domLfin_{p_{j}+1}) \rangle$$
be one of the logics in Theorem~\ref{teorX} given by a product of matrix logics, for some  $\{p_{1},\ldots, p_{j}\}\subseteq X$. By definition of $X$, for every $1 \leq s \leq j$ there exists  $1 \leq k_s < p_s$ such that $k_s/p_s \leq 1/2$ and $\langle \algLfin_{p_{s}+1}, F_{\frac{i}{n}}\cap \domLfin_{p_{s}+1}\rangle = \mathsf{L}^{k_s}_{p_s}$. This means that
$L=\mathsf{L}^{k_1}_{p_1} \times \cdots \times \mathsf{L}^{k_j}_{p_j}$.
Using again~\cite[Proposition~6.3]{CoEsGiGo} it follows that each $\mathsf{L}^{k_s}_{p_s}$ is an  {\bf LFI} w.r.t. $\neg$, with consistency operator $\circ^{k_s}_{p_s}$ defined as above. It is immediate to see that ${\sim}^i_n$ restricted to $\domLfin_{p_{s}+1}$ coincides with ${\sim}^{k_s}_{p_s}$, and so $\circ^i_n$ restricted to $\domLfin_{p_{s}+1}$ coincides with $\circ^{k_s}_{p_s}$, for every $1 \leq s \leq j$. Therefore $L$ is an  {\bf LFI} w.r.t. $\neg$, with consistency operator given by $\circ\alpha :=\circ^i_n\alpha$.

Indeed, it is clear that $\varphi, \neg\varphi, \circ\varphi \vdash_L
\psi$ for every formulas $\varphi,\psi$. Let $q$ and $r$ be two
different propositional variables, and  let $e$ be an evaluation over
$L$ such that $e(q)=1$ and $e(r)=0$. This ensures that $q, \circ q
\nvdash_L r$. On the other hand, any evaluation $e'$ over $L$ such
that $e'(q)=e'(r)=0$ guarantees that $\neg q, \circ q \nvdash_L r$.
Hence, $L$ is an {\bf LFI} w.r.t. $\neg$ and $\circ$
(recall the definition of {\bf LFI}s in~\cite{{car:01, CCM, CC16}}).
\end{proof}

\section{Conclusions}

In this paper the G\"odel fuzzy logic G expanded with an involutive negation $\sim$, denoted G$_\sim$, together with  its finite-valued extensions G$_{n\sim}$, were studied from the point of view of paraconsistency. In order to do this,  the respective degree-preserving companions $\Gsimleq$  and $\Gnsimleq$ were analyzed given that, in contrast to  G$_\sim$ and G$_{n\sim}$, these logics are $\sim$-paracon\-sis\-tent. Observe that G coincides with G$^\leq$, since $\rm{G}$ satisfies the deduction-detachment theorem; hence, the addition of an involutive negation $\sim$ to $\rm{G}$ allows to develop such kind of study. The question of determining the lattice of intermediate logics between $\Gsimleq$ and CPL, as well as the logics between $\Gnsimleq$ and CPL, was addressed. After introducing the  concept of saturated paraconsistent logic, which is weaker than the notion of ideal paraconsistency, it was shown that  there are only three saturated paraconsistent logics {between $\Gnsimleq$ and CPL}, two of them ($\mathsf{J}_3$ and $\mathsf{J}_4$) being in fact ideal paraconsistent and the other (namely, $\mathsf{J}_3 \times \mathsf{J}_4$) being saturated but not ideal.
Finally, the study of finite valued {\L}ukasiewicz logic we started in~\cite{CoEsGiGo}  has been taken up again, in order to find additional  interesting examples of saturated but not ideal  paraconsistent logics.

As for future work we aim at performing similar studies for other locally finite fuzzy logics, in particular for the Nilpotent Minimum logic (NM) \cite{Esteva-Godo:Monoidal}, that combines and shares many features of both G\"odel and {\L}ukasiewicz logics. It is indeed logically equivalent to G\"odel logic with involution when NM is expanded with the Baaz-Monteiro operator $\Delta$.

\subsection*{Acknowledgments} The authors are indebted to two anonymous reviewers for their careful and insightful suggestions and remarks that have definitively helped to improve the paper. Gispert acknowledges partial support by the Spanish FEDER/MINECO projects (MTM2016-74892 and MDM-2014-044) and grant 2017-SGR-95 of
Generalitat de Catalunya. Esteva and Godo acknowlwdge partial support by the Spanish FEDER/MINECO project TIN2015-71799-C2-1-P.


\begin{thebibliography}{9}


\bibitem{ArieliAvron}
O. Arieli and A. Avron. Four-Valued Paradefinite Logics. {\em Studia
Logica}, 105(6):1087--1122,  2017.


\bibitem{ArieliAZ10}
O. Arieli, A. Avron, and A. Zamansky.
\newblock Maximally paraconsistent three-valued logics.
\newblock In F. Lin et al. (eds.),
  {\em Principles of Knowledge Representation and Reasoning: Proc. of the  Twelfth International Conference}, (KR 2010), {AAAI} Press,
  pp. 9--13, 2010.

\bibitem{ArieliAZ11}
O. Arieli, A. Avron, and A. Zamansky.
\newblock Maximal and premaximal paraconsistency in the framework of
  three-valued semantics.
\newblock {\em Studia Logica}, 97(1):31--60, 2011.


\bibitem{ArieliAZ11a}
O. Arieli, A. Avron, and A. Zamansky. Ideal paraconsistent logics.
 {\em Studia Logica}, 99(1-3):31--60, 2011.

\bibitem{Avr:05}
A. Avron. Non-deterministic Matrices and Modular Semantics of Rules. In: J.-Y. Beziau (ed.), {\em Logica Universalis}, pp. 149--167. Birkh\"auser Basel,  2005.


\bibitem{Avr:07}
A. Avron. Non-deterministic Semantics for Logics with a Consistency Operator. {\em International Journal of Approximate Reasoning}, 45:271--287, 2007.


\bibitem{Avr:09}
A. Avron. Modular Semantics for Some Basic Logics of Formal Inconsistency. In: W. A. Carnielli, M. E. Coniglio, I. M. L. D' Ottaviano (eds), {\em The Many Sides of Logic}, pp. 15--26. {\em Studies in Logic} vol.  21, College Publications, 2009.

\bibitem{Avron16} A. Avron.
Paraconsistent fuzzy logic preserving non-falsity. {\em 
Fuzzy Sets and Systems} 292:75--84, 2016.


\bibitem{Avron17} A. Avron.
Self-Extensional Three-Valued Paraconsistent Logics. {\em Logica Universalis} 11(3):297--315, 2017.

\bibitem{Avron19}
A. Avron.
Paraconsistency and the need for infinite semantics.
{\em Soft Computing}, 23:2167--2175, 2019.


\bibitem{AvronAZ10}
A. Avron, O. Arieli, and A. Zamansky.
\newblock On strong maximality of paraconsistent finite-valued logics.
\newblock In {\em Proceedings of the 25th Annual {IEEE} Symposium on Logic in
  Computer Science, {LICS} 2010, 11-14 July 2010, Edinburgh, United Kingdom},
  p. 304--313. {IEEE} Computer Society, 2010.


 \bibitem{AAZbook}
A. Avron, O. Arieli,  and A. Zamansky. {\em Theory of effective propositional paraconsistent logics}.  Studies in Logic, Volume 75. College Publications, 2018.


\bibitem{avr:lev:01}
A. Avron and I. Lev.
Canonical propositional {G}entzen-type systems.
In: {\em  Proceedings of the 1st {I}nternational {J}oint {C}onference on {A}utomated {R}easoning ({IJCAR} 2001)}, volume 2083 of {\em LNAI}, pages 529--544.
  Springer Verlag, 2001.


\bibitem{avr:zam:11}
A. Avron and A. Zamansky. Non-deterministic semantics for logical
systems--a survey. In: D. M. Gabbay and F. Guenthner, editors, {\em Handbook of Philosophical Logic (2nd. edition)}, volume 16, pages 227--304. Springer, 2011.


\bibitem{BeCiHa}  L. B\v{e}hounek, P. Cintula and P. H\'ajek. Introduction to Mathematical Fuzzy Logic. Chapter I in P. Cintula, P. H\'ajek, C. Noguera (eds.).
\newblock {\em Handbook of Mathematical Fuzzy Logic} -- volume 1,
\newblock Studies in Logic, Mathematical Logic and Foundations, vol.~37, College Publications, London, pp- 1-101, 2011. 


\bibitem{modalbook}
P. Blackburn, M. de Rijke and Y. Venema. {\em Modal Logic}, Cambridge University Press, 2002. 

\bibitem{blok:pig:01}
W.J. Blok and D. Pigozzi. Abstract algebraic logic and the deduction theorem.
Preprint (2001) Available at\\ {\tt\small http://www.math.iastate.edu/dpigozzi/papers/aaldedth.pdf}

\bibitem{Blok-Pigozzi:DeductionTheorems}
W.J. Blok and D. Pigozzi. Local deduction theorems in algebraic logic. In {\em Algebraic logic ({B}udapest, 1988)}, volume~54 of {\em Colloq. Math. Soc. J{\'{a}}nos Bolyai}, pp. 75--109. North-Holland,
  Amsterdam, 1991.

\bibitem{BEFGGTV} F. Bou, F. Esteva, J. M. Font, A. Gil, L. Godo, A. Torrens and V. Verd\'{u}. Logics preserving degrees of truth from varieties of residuated lattices. \emph{Journal of Logic and Computation}, 19(6):1031--1069, 2009.

\bibitem{bursan81}S. Burris and H. P.  Sankappanavar. {\em A~Course
in Universal Algebra}, Springer-Verlag, 1981.

\bibitem{CC16} W.~A. Carnielli and M.~E. Coniglio. {\em Paraconsistent Logic: Consistency, contradiction and negation}. Volume 40 of {\em Logic, Epistemology, and the Unity of Science} series. Springer, 2016.

\bibitem{CCM} W.~A. Carnielli, M.~E. Coniglio and J. Marcos. {\em Logics of Formal Inconsistency}. In D.~M. Gabbay and F. Guenthner, editors, {\em Handbook of Philosophical Logic (2nd. edition)}, volume 14, pages 1--93. Springer, 2007.


\bibitem{car:01}
W.~A. Carnielli and J. Marcos. {\em A taxonomy of {C}-systems}. In W.~A. Carnielli, M.~E. Coniglio and I.~M.~L. D'Ottaviano, editors,  {\em Paraconsistency: The Logical Way to the Inconsistent.
  Proceedings of the 2nd World Congress on Paraconsistency (WCP 2000)}, volume  228 of {\em Lecture Notes in Pure and Applied Mathematics}, pages 1--94. Marcel Dekker, New York, 2002.

\bibitem{CoEsGo:14}
M.E. Coniglio, F. Esteva, and L. Godo. Logics of formal inconsistency arising from systems of fuzzy logic. {\em Logic Journal of the IGPL}, 22(6):880--904, 2014.


\bibitem{CoEsGo}  M. E. Coniglio, F. Esteva and L. Godo.
On the set of intermediate logics between truth and degree preserving {\L}ukasiewicz logic
 \emph{Logical Journal of the IGPL} 24(3):288--320, 2016.

\bibitem{CoEsGiGo} M. E. Coniglio, F. Esteva, J. Gispert, L. Godo.
Maximality in finite-valued {\L}ukasiewicz Logics defined by order filters
\emph{Journal of Logic and Computation} 29(1):125--156, 2019.

\bibitem{EEFGN} R.C. Ertola, F. Esteva, T. Flaminio, L. Godo, C. Noguera.
Paraconsistency properties in degree-preserving fuzzy logics. {\em Soft Computing} 19(3):531--546, 2015.

\bibitem{Esteva-Godo:Monoidal}
Francesc Esteva and Llu{\'{\i}}s Godo.
Monoidal t-norm based logic: Towards a logic for left-continuous
  t-norms.
{\em Fuzzy Sets and Systems}, 124(3):271--288, 2001.

\bibitem{EsGoHaNa} F. Esteva , L. Godo, P. H\'{a}jek, M. Navara.
Residuated fuzzy logics with an involutive negation.
 \emph{Archive for Mathematical Logic}, 39:103--124, 2000.

\bibitem{EsGoMa} F. Esteva , L. Godo, E. Marchioni.
Fuzzy Logics with Enriched Language. In: P. Cintula, P. H\'ajek and C. Noguera (eds.), {\em Handbook of Mathematical Fuzzy Logic, Vol. II}, Chap. VIII, pp. 627--711. Volume 38 of {\em Studies in Logic, Mathematical Logic and Foundations}. College Publications, 2011.

\bibitem{Font} J.M. Font. Taking degrees of truth seriously. {\em Studia Logica} 91(3): 383-406, 2009.

\bibitem{Font-GTV} J.M. Font, A. Gil, A. Torrens and V. Verd\'u. On the infinite-valued {\L}ukasiewicz logic that
preserves degrees of truth. {\em Archive for Mathematical Logic}, 45, 839-868, 2006.


\bibitem{GT14}
J.Gispert and A. Torrens.  Locally finite quasivarieties of {MV}-algebras.
{\em ArXiv}, pp: 1--14, 2014. Online DOI: http://arxiv.org/abs/1405.7504.

\bibitem{Hajek98} {P. H\'ajek}. {\em Metamathematics of Fuzzy Logic}. Volume 4 of {\em Trends in Logic}. Kluwer Academic Publishers, 1998.

\bibitem{odin:08}
S.~P. Odintsov.
{\em Constructive Negations and Paraconsistency}. Volume~26 of {\em
  Trends in Logic}. Springer, 2008.


\bibitem{RibCon} M.M. Ribeiro and M.E. Coniglio. Contracting Logics. In: L. Ong and R. de Queiroz (eds), {\em Logic, Language, Information and Computation. WoLLIC 2012}, pp. 268--281. Lecture Notes in Computer Science, vol 7456. Springer, Berlin,  2012.

\bibitem{WansOdin} H. Wansing and S.P. Odintsov.
On the Methodology of Paraconsistent Logic. In: H. Andreas and P. Verd\'ee (eds), {\em Logical Studies of Paraconsistent Reasoning in Science and Mathematics}, pp. 175--204. Volume 45 of {\em Trends in Logic}. Springer, 2016.


\end{thebibliography}
\end{document}